\documentclass[11pt,reqno]{amsart}
\usepackage{amsmath, amsthm, amssymb}
\usepackage{enumitem}
\usepackage[foot]{amsaddr}

\usepackage{fullpage}
\usepackage[many]{tcolorbox}
\usepackage{xcolor}
\usepackage{wasysym}
\usepackage{mathtools}
\usepackage{scalerel}
\usepackage{bbm}
\usepackage[all]{xy}
\usepackage{mathrsfs}

\usepackage{tikz}
\usetikzlibrary{arrows,backgrounds,patterns.meta}
\usetikzlibrary{positioning,shadings,cd}
\usetikzlibrary{shapes}
\usetikzlibrary{backgrounds}
\usetikzlibrary{decorations,decorations.pathreplacing,decorations.markings,decorations.pathmorphing}
\tikzstyle{snake}=[decorate, decoration={snake, segment length=1mm, amplitude=.5mm}]
\usetikzlibrary{fit,calc,through}
\usetikzlibrary{external}
\newcommand{\tikzmath}[2][]
{\vcenter{\hbox{\begin{tikzpicture}[#1]#2\end{tikzpicture}}}
}
\newcommand{\roundNbox}[6]{
	\draw[rounded corners=5pt, very thick, #1] ($#2+(-#3,-#3)+(-#4,0)$) rectangle ($#2+(#3,#3)+(#5,0)$);
	\coordinate (ZZa) at ($#2+(-#4,0)$);
	\coordinate (ZZb) at ($#2+(#5,0)$);
	\node at ($1/2*(ZZa)+1/2*(ZZb)$) {#6};
}

\tikzset{super thick/.style={line width=3pt}}
\tikzstyle{far>}=[decoration={markings, mark=at position 0.75 with {\arrow{>}}}, postaction={decorate}]
\tikzstyle{mid>}=[decoration={markings, mark=at position 0.55 with {\arrow{>}}}, postaction={decorate}]
\tikzstyle{mid<}=[decoration={markings, mark=at position 0.55 with {\arrow{<}}}, postaction={decorate}]
\tikzset{super thick/.style={line width=3pt}}
\tikzstyle{far>}=[decoration={markings, mark=at position 0.75 with {\arrow{>}}}, postaction={decorate}]
\tikzstyle{mid>}=[decoration={markings, mark=at position 0.55 with {\arrow{>}}}, postaction={decorate}]
\tikzstyle{mid<}=[decoration={markings, mark=at position 0.55 with {\arrow{<}}}, postaction={decorate}]
\tikzstyle{knot}=[preaction={super thick, white, draw}]
\tikzstyle{coupon}=[draw, very thick, rectangle, rounded corners=5pt]
\tikzset{Rightarrow/.style={double equal sign distance,>={Implies},->},
triplecd/.style={-,preaction={draw,Rightarrow}},
quadruplecd/.style={preaction={draw,Rightarrow,
shorten >=0pt
},
shorten >=1pt,
-,double,double
distance=0.2pt}}
\tikzset{
    tripleline/.style args={[#1] in [#2] in [#3]}{
        #1,preaction={preaction={draw,#3},draw,#2}
    }
}
\tikzstyle{triple}=[tripleline={[line width=.15mm,black] in
      [line width=.7mm,white] in
      [line width=1mm,black]}] 
\tikzset{
    quadrupleline/.style args={[#1] in [#2] in [#3] in [#4]}{
        #1,preaction={preaction={preaction={draw,#4},draw,#3}, draw,#2}
    }
}
\tikzstyle{quadruple}=[quadrupleline={[line width=.3mm,white] in
      [line width=.6mm,black] in
      [line width=1.2mm,white] in
      [line width=1.5mm,black]}]

\definecolor{violet}{RGB}{148,0,211}
\definecolor{DarkGreen}{RGB}{0,150,0}
\definecolor{rufous}{HTML}{A81C07}

\usepackage[plainpages=false,hypertexnames=false,pdfpagelabels]{hyperref}
\definecolor{medium-blue}{rgb}{0,0,.8}
\hypersetup{colorlinks, linkcolor={purple}, citecolor={medium-blue}, urlcolor={medium-blue}}
\newcommand{\arxiv}[1]{\href{http://arxiv.org/abs/#1}{\tt arXiv:\nolinkurl{#1}}}
\newcommand{\arXiv}[1]{\href{http://arxiv.org/abs/#1}{\tt arXiv:\nolinkurl{#1}}}

\newcommand{\doi}[1]{\href{http://dx.doi.org/#1}{{\tt DOI:#1}}}

\newcommand{\googlebooks}[1]{(preview at \href{https://books.google.com/books?id=#1}{google books})}

\DeclareMathOperator{\End}{End}

\DeclareMathOperator{\id}{id}
\DeclareMathOperator{\im}{im}

\DeclareMathOperator{\Irr}{Irr}

\DeclareMathOperator{\op}{op}

\DeclareMathOperator{\Tr}{Tr}
\DeclareMathOperator{\tr}{tr}

\newcommand{\set}[2]{\left\{#1 \middle| #2\right\}}

\def\semicolon{;}
\def\applytolist#1{
    \expandafter\def\csname multi#1\endcsname##1{
        \def\multiack{##1}\ifx\multiack\semicolon
            \def\next{\relax}
        \else
            \csname #1\endcsname{##1}
            \def\next{\csname multi#1\endcsname}
        \fi
        \next}
    \csname multi#1\endcsname}

\def\calc#1{\expandafter\def\csname c#1\endcsname{{\mathcal #1}}}
\applytolist{calc}QWERTYUIOPLKJHGFDSAZXCVBNM;
\def\bbc#1{\expandafter\def\csname bb#1\endcsname{{\mathbb #1}}}
\applytolist{bbc}QWERTYUIOPLKJHGFDSAZXCVBNM;
\def\bfc#1{\expandafter\def\csname bf#1\endcsname{{\mathbf #1}}}
\applytolist{bfc}QWERTYUIOPLKJHGFDSAZXCVBNM;
\def\sfc#1{\expandafter\def\csname s#1\endcsname{{\sf #1}}}
\applytolist{sfc}QWERTYUIOPLKJHGFDSAZXCVBNM;
\def\fc#1{\expandafter\def\csname f#1\endcsname{{\mathfrak #1}}}
\applytolist{fc}QWERTYUIOPLKJHGFDSAZXCVBNM;
\def\rmc#1{\expandafter\def\csname rm#1\endcsname{{\mathrm #1}}}
\applytolist{rmc}QWERTYUIOPLKJHGFDSAZXCVBNM;
\def\scrc#1{\expandafter\def\csname scr#1\endcsname{{\mathscr #1}}}
\applytolist{scrc}QWERTYUIOPLKJHGFDSAZXCVBNM;

\numberwithin{equation}{section}

\theoremstyle{plain}
\newtheorem{thm}[equation]{Theorem}
\newtheorem*{thm*}{Theorem}
\newtheorem{cor}[equation]{Corollary}
\newtheorem{lem}[equation]{Lemma}
\newtheorem{prop}[equation]{Proposition}

\newtheorem{question}[equation]{Question}
\newtheorem*{claim*}{Claim}

\newtheorem{thmalpha}{Theorem}

\theoremstyle{definition}
\newtheorem{defn}[equation]{Definition}

\newtheorem{assume}[equation]{Assumption}

\newtheorem{axiom}[equation]{Axiom}
\newtheorem*{trick*}{Trick}
\newtheorem{construction}[equation]{Construction}
\newtheorem{nota}[equation]{Notation}
\newtheorem{fact}[equation]{Fact}
\newtheorem{facts}[equation]{Facts}

\newtheorem{ex}[equation]{Example}
\newtheorem{exs}[equation]{Examples}

\newtheorem{rem}[equation]{Remark}

\newtheorem{warn}[equation]{Warning}

\title{Local topological order, Haag duality, and reflection positivity}
\date{\today}
\begin{document}
\author{Pieter Naaijkens$^1$}
\address{$^1$ School of Mathematics, Cardiff University, Cardiff, CF24 4AG, United Kingdom}
\author{David Penneys$^2$}
\author{Daniel Wallick$^2$}
\address{$^2$ Department of Mathematics, The Ohio State University, Columbus, OH 43210, USA}

\begin{abstract}
In our previous article [\arxiv{2307.12552}], we introduced local topological order (LTO) axioms for abstract quantum spin systems which allow one to access topological order via a boundary algebra construction.
Using the LTO axioms, we produced a canonical pure state on the quasi-local algebra, which gives a net of von Neumann algebras associated to a poset of cones in $\mathbb{R}^n$.
In this article, motivated by [\arxiv{2509.23734}], we introduce an axiom for LTOs which ensures Haag duality for cone-like regions using Tomita-Takesaki theory.
We prove this axiom is satisfied for all known topologically ordered commuting projector models.
We thus get an independent proof of Haag duality for the Levin-Wen string net models originally proved in [\arxiv{2509.23734}].
We also give a reflection positivity axiom for LTOs, connecting to the recent article [\arxiv{2510.20662}].
We again prove this axiom is satisfied for all known topologically ordered commuting projector models about some $\mathbb{Z}/2$-reflection symmetry.
\end{abstract}
\maketitle
\tableofcontents

\section{Introduction}

Topologically ordered 2D quantum spin systems 
exhibit emerging anyons that form a braided tensor category \cite{MR1951039, cond-mat/0506438, PhysRevB.71.045110}. 
Recently, operator algebraic approaches, inspired by the Dopicher--Haag--Roberts approach to algebraic quantum field theory \cite{MR0297259, MR334742}, have been used to obtain the braided tensor category of anyons for many of these models on infinite planar lattices \cite{MR2804555, MR3135456, MR3426207, MR4362722, 2306.13762, 2310.19661, MR4998297, 2511.21521, 2603.01936}.
The approach used by these papers can be formulated using the language of a net of von Neumann algebras, where a von Neumann algebra is associated to each cone in the lattice \cite{MR4927814}, leading to a category of superselection sectors.
In order to prove that there is a braided tensor structure on this category, one uses a technical condition called \emph{Haag duality}, or a generalization thereof \cite{MR4362722}, which states that the von Neumann algebras for each cone and its complement are each other's commutants.
Haag duality also has a quantum information theoretic interpretation \cite{MR5026085}. 
In \cite{2509.23734}, it was shown that a large class of models, specifically string-net models \cite{PhysRevB.71.045110}, satisfy Haag duality, building on previous work showing Haag duality for the Toric Code \cite{MR2956822} and abelian Quantum Double Model \cite{MR3426207}. 

In this article, we provide a sufficient condition for Haag duality that is easy to check for all known topologically ordered commuting projector models. 
To do so, we employ the \emph{local topological order} (LTO) axioms of \cite{MR4945955}, which characterize 
topological order locally using operator algebraic conditions.
The setup of the LTO axioms is a net of local algebras $R\mapsto \fA(R)$ equipped with a net of projections $p_R\in \fA(R)$ satisfying $p_S\leq p_R$ whenever $R\subseteq S$. 
In 2D, we obtain a 1D net of \emph{boundary algebras} from the LTO axioms after cutting the system using a hyperplane, 
which can be used to implement a topological holographical principle in the sense of \cite{PhysRevB.107.155136,2310.05790}.
For the known commuting projector models, the net of boundary algebras is identical to the net of interaction algebras \cite{MR1816609}, and the \emph{DHR bimodules} recover the bulk topological order \cite{MR4814692,MR4945955}.

We give a \emph{Haag duality axiom for LTOs} inspired by \cite{2509.23734}, which relates the boundary algebras on each side of the cut, which is again satisfied by all known topologically ordered commuting projector models.
The axiom, in the setting of a quantum spin system, can be loosely summarized as follows.
(We refer the reader to \ref{LTO:HD} below for the precise statement.)

\begin{itemize}
\item 
For any cone $\Lambda_+$ with complement $\Lambda_-$,\footnote{Here, `cone' means \emph{cone-like region}, which is assumed to be sufficiently large with boundary homeomorphic to $\bbR$.
See \S\ref{sec:LTO} for more details.} 
suppose we have sufficiently nice large regions $R,S$, with $R$ sufficiently contained in $S$, whose boundaries intersect $\partial \Lambda_\pm$ transversely.
Then
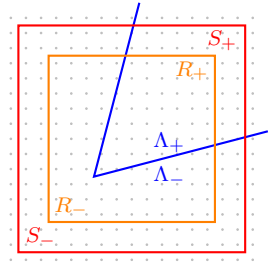
\[
p_{S_+} \fA(R_+) p_S
=
p_{S_-} \fA(R_-) p_S.
\qquad\qquad\qquad
\tikzmath{
\foreach \x in {-8,...,8}{
\foreach \y in {-8,...,8}{
\filldraw[gray!50!white] ($ .2*(\x,\y) $) circle (.01cm);
}}
\draw[thick, blue] (0.1,1.8) -- (-.5,-.5) -- (1.8,.1);
\node[blue] at (.5,-.05) {$\scriptstyle \Lambda_+$};
\node[blue] at (.5,-.5) {$\scriptstyle \Lambda_-$};
\draw[thick, red] (-1.5,-1.5) rectangle (1.5,1.5);
\node[red] at (1.2,1.3) {$\scriptstyle S_+$};
\node[red] at (-1.2,-1.3) {$\scriptstyle S_-$};
\draw[thick, orange] (-1.1,-1.1) rectangle (1.1,1.1);
\node[orange] at (.8,.9) {$\scriptstyle R_+$};
\node[orange] at (-.8,-.9) {$\scriptstyle R_-$};
}
\]
\end{itemize}
This equality says that when restricted to the local ground state space for $S$, the actions of the boundary algebras on either of the $\pm$ side of the cut yield the same `topological error excitations' on the physical boundary cut.

Our main result is that \ref{LTO:HD} for LTOs implies Haag duality for the associated net of cone von Neumann algebras in the canonical ground state.
That is, by \cite[\S2.3]{MR4945955}, the LTO axioms give us a canonical pure state $\psi$ on $\fA$ satisfying $\psi(p_R)=1$ for all regions $R$, and we get a net of cone von Neumann algebras by $\cA(\Lambda)\coloneqq \fA(\Lambda)''$ on $L^2(\fA,\psi)$.
The von Neumann completion of the boundary algebra along $\partial \Lambda$ is then given by $\cB(\Lambda)\coloneqq p_\Lambda \cA(\Lambda)p_\Lambda$.
We get an independent proof of Haag duality for the known topologically ordered commuting projector models.

\begin{thmalpha}
\label{thm:Main}
Let $(\fA, p)$ be an LTO with finite dimensional local algebras satisfying the \ref{LTO:HD} axiom. 
Assume that $\psi$ is faithful on $\cB(\Lambda_\pm)$,
and $p_{\Lambda_\pm}$ have central support 1 in $\cA(\Lambda_\pm)$.
Then $\cA(\Lambda_+)'=\cA(\Lambda_-)$ on $L^2(\fA,\psi)$.
\end{thmalpha}

The faithfulness assumption in Theorem \ref{thm:Main} above allows us to show that both boundary algebras $\cB(\Lambda_\pm)$ are in standard form on $p_{\Lambda_+}p_{\Lambda_-}L^2(\fA,\psi)$, so we may employ \emph{Tomita--Takesaki theory} to get the desired result.
Again, this faithfulness has been verified for the known topologically ordered commuting projector models.
Additionally, it implies a strengthened injectivity axiom for LTOs used in \cite{2507.03201} to connect their frustration-free setup to the LTO machinery from \cite{MR4945955}.
The assumption that $p_{\Lambda_\pm}$ have central support 1 in $\cA(\Lambda_\pm)$ is automatic for quantum spin systems by purity of $\psi$, which ensures $\cA(\Lambda_\pm)$ are factors. 
However, our treatment also applies to \emph{abstract} spin systems satisfying our assumptions, where the cone algebras are not necessarily factors \cite{2506.19969}. 

Additionally, we show that LTOs satisfy a \emph{mutual information condition} that, together with Haag duality, is used in \cite{2511.08382} to show that the \emph{approximate split property} is satisfied for the cone von Neumann algebras. 
We also remark what can be said when the cone algebras are not factors, in which case the proof in \cite{2511.08382} does not work without some additional assumptions. 

Finally, we connect our LTO axioms with the \emph{reflection positivity} setup developed in \cite{2510.20662} to show when frustration-free models automatically satisfy a local topological quantum order axiom from \cite{MR2742836}, going beyond \cite[\S6]{2507.03201}.
Reflection positivity has been highly successful in constructive quantum field theory and quantum statistical mechanics (see e.g.~\cite{MR0887102,Simon2025} and references therein).
We provide an axiom \ref{LTO:RP} for LTOs involving the reflection positive structure that, for quantum spin systems, implies our LTO Haag duality axiom \ref{LTO:HD}. 
Assuming a $\bbZ/2$-reflection about a vertical cut, the axiom \ref{LTO:RP} relates the right action of the right boundary algebra along the cut in the ground state to the 
left action of the left boundary algebra along the cut in the ground state, but twisted by $\sigma^\psi_{-i/2}$ where $t\mapsto \sigma^\psi_t$ is the modular group for the canonical state.
We again prove that known commuting projector models satisfy \ref{LTO:RP} for some $\bbZ/2$-symmetry.
In \cite{2510.20662}, the authors use reflection positivity and Osterwalder--Schrader (OS) reconstruction~\cite{MR376002} to obtain a net of algebras which is isomorphic to the net of interaction algebras, but again twisted by the modular automorphism group.
We hope that by relating our setup to the one in \cite{2510.20662}, the tools developed in this article might be useful in proving Haag duality from the reflection positive setup in \cite{2510.20662}.

\subsection*{Acknowledgments}
The authors would like to thank Anupama Bhardwaj, Corey Jones, Kyle Kawagoe, Benjamin Major, Milo Moses, Yoshiko Ogata, and David P\'erez-Garc\'ia for helpful conversations. 
David Penneys and Daniel Wallick were supported by NSF DMS-2154389.

\section{Background}

In this section, we recall the basic objects of study in this note, and we prove some new results in \S\ref{sec:ConditionalExpectations} below.
Our basic setup is that of a \emph{net of algebras}: a unital quasilocal $\rmC^*$-algebra $\fA$ together with an assignment of unital $\rmC^*$-subalgebras $\fA(R)$ of $\fA$ to a class of subregions $R$ in a lattice $\cL$.
In \cite{MR4945955}, we worked with rectangular subsets in a $\bbZ^n$ lattice, which we extend to `disk-like' regions in \S\ref{sec:LTO} below.
These $\rmC^*$-subalgebras are required to satisfy the following axioms:
\begin{enumerate}[label=\textup{(N\arabic*)}]
\item \label{it:netempty}
(vacuum) $\fA(\emptyset)=\bbC \cdot 1_\fA$,
\item
(isotony) $R \subset S$ implies $\fA(R)\subset \fA(S)$,
\item
(locality) $R \cap S = \emptyset$ implies $[\fA(R), \fA(S)]=0$ inside $\fA$, and
\item \label{it:netdense}
(nondegeneracy) $\bigcup_{R \subset \cL} \fA(R)$ is dense in $\cL$.
\end{enumerate}

A typical example is a quantum spin system on a lattice, where for each site $x$ of the lattice we have $\fA(\{x\}) \cong M_d(\mathbb{C})$, and for $R \subset \cL$ finite, $\fA(R)$ is defined by taking the tensor product of the algebras for each site in $R$.
We will however not assume our local algebras are tensor products of single-site algebras unless explicitly stated.
Such a system is sometimes called an \emph{abstract quantum spin system} and we reserve the terminology `quantum spin system' for the case where the local algebras are tensor products of full matrix algebras as above.

\subsection{Local topological order axioms}
\label{sec:LTO}

We recall the local topological order axioms introduced in~\cite{MR4945955}.
The axioms are a generalization of the LTQO axioms introduced in~\cite{MR2742836,MR2842961} that allow us to study what happens at the \emph{boundary} of a region (rather than the bulk considered in the LTQO axioms).
A \emph{net of projections} for the net of algebras $\fA$ is an assignment of an orthogonal projection $p_R \in \fA(R)$ for every region $R$, such that $R \subset S$ implies $p_S \leq p_R$ in $\fA(S)$.
In applications the projections $p_S$ are typically the local ground state projections of a frustration-free Hamiltonian~\cite{MR1158756, MR2742836, MR4700364}.

\begin{nota}
For two rectangles $R,S$ and $s>0$, we say $R$ is \emph{completely surrounded} by $S$, denoted $R \ll_s S$, if for every point $\ell\in R$, the $s$-ball in the $\ell^\infty$ metric is contained in $S$.
Observe that $R \ll_sS$ if and only if both $R\subset S$ and the distance from $\partial R$ to $\partial S$ is at least $s$, where $\partial$ denotes taking the boundary.

We say $R$ is \emph{(weakly) surrounded} by $S$, denoted by $R \Subset_sS$, if $R\subseteq S$, $I\coloneqq \partial R \cap \partial S$ consists of a single non-empty face of $R$, and 
every point of $S\setminus R$ is contained in some $s$-cube which is completely contained in $S\setminus R$.
This means for every other face $F$ of $R$ except $I$, the distance from $F$ to $\partial S$ in the direction of a normal vector pointing away from $S$ is at least $s$.

We represent $R \ll_sS$ and $R \Subset_s S$ on the left and right-hand sides of the following cartoon.
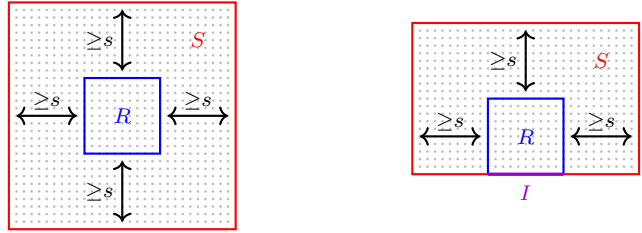
\[
\tikzmath{
\foreach \x in {-9,...,19}{
\foreach \y in {-9,...,19}{
\filldraw[gray!50!white] ($ .1*(\x,\y) $) circle (.01cm);
}}
\draw[thick, blue] (0,0) rectangle (1,1);
\node[blue] at (.5,.5) {$\scriptstyle R$};
\draw[thick, red] (-1,-1) rectangle (2,2);
\node[red] at (1.5,1.5) {$\scriptstyle S$};
\draw[thick, <->] (1.1,.5) -- (1.9,.5);
\node at (1.5,.7) {$\scriptstyle \geq s$};
\draw[thick, <->] (-.1,.5) -- (-.9,.5);
\node at (-.5,.7) {$\scriptstyle \geq s$};
\draw[thick, <->] (.5,1.1) -- (.5,1.9);
\node at (.2,1.5) {$\scriptstyle \geq s$};
\draw[thick, <->] (.5,-.1) -- (.5,-.9);
\node at (.2,-.5) {$\scriptstyle \geq s$};
}
\qquad\qquad\qquad
\tikzmath{
\foreach \x in {-9,...,19}{
\foreach \y in {1,...,19}{
\filldraw[gray!50!white] ($ .1*(\x,\y) $) circle (.01cm);
}}
\draw[thick, blue] (0,0) rectangle (1,1);
\node[blue] at (.5,.5) {$\scriptstyle R$};
\draw[thick, red] (-1,0) rectangle (2,2);
\node[red] at (1.5,1.5) {$\scriptstyle S$};
\draw[very thick, violet] (0,0) --node[below]{$\scriptstyle I$} (1,0);
\draw[thick, <->] (1.1,.5) -- (1.9,.5);
\node at (1.5,.7) {$\scriptstyle \geq s$};
\draw[thick, <->] (-.1,.5) -- (-.9,.5);
\node at (-.5,.7) {$\scriptstyle \geq s$};
\draw[thick, <->] (.5,1.1) -- (.5,1.9);
\node at (.2,1.5) {$\scriptstyle \geq s$};
}
\]
\end{nota}

Whenever $R \Subset_s S$, we define the $*$-algebra
\[
\fB(R \Subset_s S)
\coloneqq
\set{x p_S}{x\in p_R\fA(R)p_R \text{ and }xp_{\widehat{S}}=p_{\widehat{S}}x\text{ whenever }R\Subset_s\widehat{S}\text{ and }\partial R\cap \partial S = \partial R\cap\partial \widehat{S}}.
\]

\begin{defn}
A \emph{local topological order} (LTO) is a pair $(\fA,p)$ where $\fA$ is a net of algebras equipped with a net of projections that satisfies the following axioms, for sufficiently large regions $R,S$,\footnote{\label{footnote:SufficientlyLarge}An LTO comes with a `sufficient largeness' global constant $r>0$.
We say $R$ is \emph{sufficiently large} if for every $v\in R$, there is an $r$-cube $C\subset R$ containing $v$.} 
where $s> 0$ is some fixed global constant.
\begin{enumerate}[label=\textup{(LTO\arabic*)}]
\item
\label{LTO:QECC}
Whenever $R \ll_s S$, $p_S \fA(R)p_S = \bbC \cdot p_S$.
\item
\label{LTO:Boundary}
Whenever $R \Subset_s S$, $p_S \fA(R)p_S =\fB(R \Subset_s S)\cdot p_S$.
\item
\label{LTO:Surjective}
Whenever $R_1\subset R_2 \Subset_s S$ with $\partial R_1\cap \partial S=\partial R_2\cap \partial S$, $\fB(R_1 \Subset_s S)=\fB(R_2 \Subset_s S)$.
\item
\label{LTO:Injective}
Whenever $R\Subset_sS_1\subset S_2$ with  $\partial R\cap\partial S_1=\partial R\cap\partial S_2$, if $x\in \fB(\partial R\cap\partial S_1)$ such that $xp_{S_2}=0$, then $x=0$.
\end{enumerate}
\end{defn}

The final three axioms may be loosely summarized by saying that $\fB(R\Subset_s S)$ really only depends on the sites in $\cL$ near $I$, and that there is a single algebra $\fB(I)$ that plays the role of $\bbC$ in \ref{LTO:QECC} when $R\Subset_sS$ instead of $R\ll_sS$.
For a rigorous version of this summary, we refer the reader to \cite[(2.14) and Prop.~2.15]{MR4945955}.

\begin{exs}
\label{exs:KnownTopologicalCommutingProjectorModels}
Almost all known topologically ordered commuting projector models are known to satisfy \ref{LTO:QECC}--\ref{LTO:Injective}, including
Kitaev's toric code \cite{MR4945955},
Kitaev's quantum double \cite{MR4814524},
Levin-Wen \cite{MR4945955}, and
Walker-Wang \cite{2506.19969}.
\end{exs}

Our results hold in far more generality than just using rectangles.
For our $\bbZ^n$ lattice, we consider regions $R$ that are `disk-like' in the sense that their boundaries are `sphere-like.'
Precisely, we say this as follows.

\begin{defn}
Under the usual metric embedding $\bbZ^n\subset \bbR^n$, for each $z\in \bbZ^n$, write $[z]$ for its closed $\ell^\infty$-ball (in $\mathbb{R}^n$) of radius 1/2 and $[\Lambda]\coloneqq \bigcup_{z\in\Lambda} [z]$.
A region $R$ in our $\bbZ^n$ lattice is called \emph{disk-like} if its boundary $\partial[\Lambda]$ in $\bbR^n$ is homeomorphic to $S^{n-1}$.
Below is a cartoon for $n=2$.
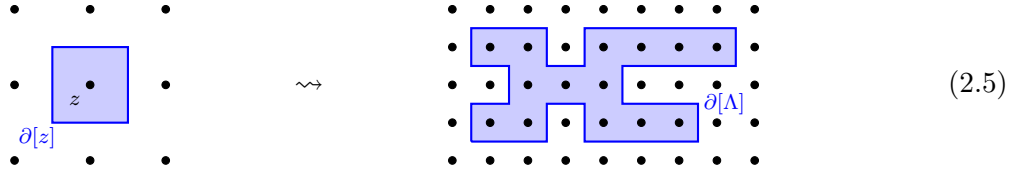
\begin{equation}
\label{eq:DisklikeRegions}    
\tikzmath{
\filldraw[blue, thick, fill=blue!20] (.5,.5) rectangle (1.5,1.5);
\foreach \x in {0,1,2}{
\foreach \y in {0,1,2}{
\filldraw (\x,\y) circle (.05cm);
}}
\node at (.8,.8) {$\scriptstyle z$};
\node[blue] at (.3,.3) {$\scriptstyle \partial[z]$};
}
\qquad\qquad\rightsquigarrow\qquad\qquad
\tikzmath{
\filldraw[blue, thick, fill=blue!20] (-.25,.25) -- (-.25, .75) -- (.25, .75) -- (.25, 1.25) -- (-.25, 1.25) -- (-.25, 1.75) -- (.25,1.75) -- (.75, 1.75) -- (.75, 1.25) -- (1.25,1.25) -- (1.25,1.75) -- (3.25,1.75) -- (3.25,1.25) -- (1.75,1.25) -- (1.75, .75) -- (2.75, .75) -- (2.75,.25) -- (1.25,.25) -- (1.25,.75) -- (.75,.75) -- (.75,.25) -- (-.25,.25);
\foreach \x in {-1,0,1,...,7}{
\foreach \y in {0,1,2,3,4}{
\filldraw ($ .5*(\x,\y) $) circle (.05cm);
}}
\node[blue] at (3.1,.75) {$\scriptstyle \partial[\Lambda]$};
}
\end{equation}
We write $\partial\Lambda$ for the set of points in $z\in \Lambda$ such that $z\in\partial[\Lambda]$.
\end{defn}

\begin{rem}
In 2D, disk-like regions generalize the `comb-like regions' of \cite[Def.~3.4.3]{2509.23734}. 
A comb-like region is always disk-like \cite[Rem.~3.4.5]{2509.23734}, but
the example on the right-hand side of \eqref{eq:DisklikeRegions} is disk-like but not comb-like.
\end{rem}

For disk-like regions we get the obvious notion of $R \ll_s S$:
$R\subset S$ 
and $\partial[R]$ is at least $\ell^\infty$-distance $s$ away from $\partial[S]$.
To define weakly surrounding regions $R\Subset_s S$, we may require that 
\begin{equation}
\label{eq:bdinterval}
I\coloneqq \partial R \cap \partial S
\end{equation}
is contained in a codimension one hyperplane $\cH$, in which it is also disk-like, i.e., $\partial[I]\subset \cH$ is homeomorphic to $S^{n-2}$,
and
every point of $S\setminus R$ is contained in some $s$-cube which is completely contained in $S\setminus R$.

\begin{exs}
One can similarly verify the axioms \ref{LTO:QECC}--\ref{LTO:Injective} for these more general disk-like regions (which are sufficiently large as in Footnote \ref{footnote:SufficientlyLarge}) for the commuting projector models mentioned in Examples \ref{exs:KnownTopologicalCommutingProjectorModels}.\footnote{One may also restrict attention to disk-like regions $R$ such that the complement $R^c$ is also sufficiently large, i.e., for every $v\in R^c$, there is an $r$-cube $C\subset R^c$ containing $v$.}

\end{exs}

For the LTO machinery to work properly for boundary algebras for an arbitrary cone $\Lambda$, we would like that
\begin{itemize}
\item 
$\Lambda$ and $\Lambda^c$ should be sufficiently large as in Footnote \ref{footnote:SufficientlyLarge}, 
\item 
$\Lambda$ and $\Lambda^c$ should be the unions of their sufficiently large disk-like regions, and
\item 
the boundary $\partial[\Lambda]=\partial[\Lambda^c]$ should be homeomorphic to $\bbR$ (see the left hand side below).
\end{itemize}
Unfortunately, this is not always the case, as in the right hand side below.
$$
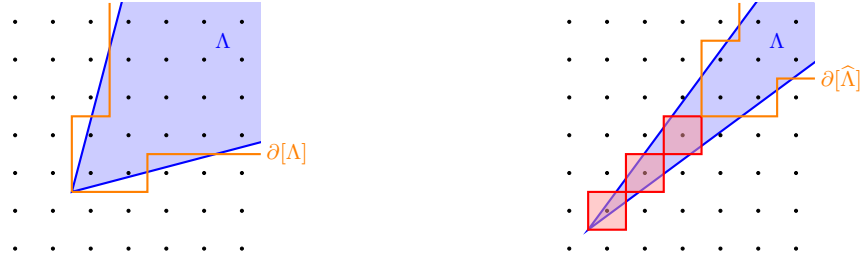
\tikzmath{
\begin{scope}
\clip (-1.75,-1.75) rectangle (1.75,1.75);
\filldraw[blue, thick, fill=blue!20] (.5,4) -- (-.75,-.75) -- (4,.5);
\end{scope}
\foreach \x in {-3,-2,...,3}{
\foreach \y in {-3,-2,...,3}{
\filldraw ($ .5*(\x,\y) $) circle (.02cm);
}}
\draw[thick,orange] (-.75,-.75) -- (.25,-.75) -- (.25,-.25) -- (1.75,-.25);
\draw[thick,orange] (-.75,-.75) -- (-.75,.25) -- (-.25,.25) -- (-.25,1.75);
\node[blue] at (1.25,1.25) {$\scriptstyle \Lambda$};
\node[orange] at (2.1,-.25) {$\scriptstyle \partial[\Lambda]$};
}
\qquad\qquad\qquad\qquad
\tikzmath{
\begin{scope}
\clip (-1.75,-1.75) rectangle (1.75,1.75);
\filldraw[blue, thick, fill=blue!20] (1.3,2.2) -- (-1.25,-1.25) -- (2.2,1.3);
\end{scope}
\foreach \x in {-3,-2,...,3}{
\foreach \y in {-3,-2,...,3}{
\filldraw ($ .5*(\x,\y) $) circle (.02cm);
}}
\fill[opacity=.5, fill=red!40] (-1.25,-1.25) rectangle (-.75,-.75);
\draw[thick, red] (-1.25,-1.25) rectangle (-.75,-.75);
\fill[opacity=.5, fill=red!40] (-.25,-.25) rectangle (-.75,-.75);
\draw[thick, red] (-.25,-.25) rectangle (-.75,-.75);
\fill[opacity=.5, fill=red!40] (-.25,-.25) rectangle (.25,.25);
\draw[thick, red] (-.25,-.25) rectangle (.25,.25);
\draw[thick,orange] (.25,.25) -- (1.25,.25) -- (1.25,.75) -- (1.75,.75);
\draw[thick,orange] (.25,.25) -- (.25,1.25) -- (.75,1.25) -- (.75,1.75);
\node[blue] at (1.25,1.25) {$\scriptstyle \Lambda$};
\node[orange] at (2.1,.75) {$\scriptstyle \partial[\widehat{\Lambda}]$};
}
$$
Equipped with our more general disk-like regions, we can define more general \emph{cone-like regions} and restrict our attention to the corresponding net of von Neumann algebras for these regions.

\begin{construction}
\label{construction:ConeLikeRegions}
For `thin' cones where $\partial[\Lambda]$ is not homeomorphic to $\bbR$, one can always remove finitely many points from $\Lambda$ to obtain a new \emph{cone-like region} $\widehat{\Lambda}$ (see the right hand side above) which is sufficiently large and such that $\partial[\widehat{\Lambda}]$ is again homeomorphic to $\bbR$.
For `fat' cones whose complements are `thin,' we may add finitely many points instead.
\end{construction}

\begin{defn}
For a cone-like region $\Lambda$, we define
$\fA(\Lambda)\coloneqq \overline{\bigcup_{R\subset \Lambda} \fA(R)}^{\|\cdot\|}$.
\end{defn}

\begin{fact}
These cone-like-regions form a poset with order-reversing involution which satisfies the conditions of \cite{MR4927814} in order to define a braided $\rmW^*$-tensor category of superselection sectors in the case that the axioms are satisfied.    
For our purposes it is thus sufficient to restrict to sufficiently large cone-like regions instead of arbitrary cones.
We will further study the net of algebras $\Lambda\mapsto \fA(\Lambda)$ for cone-like regions in \S\ref{sec:NetsOfVNA} below.
\end{fact}

\begin{warn}
In the sequel, we will often write `cone' for `cone-like region.' 
\end{warn}

\begin{rem}
Our definition of cone-like regions is morally the same as the definition of cone-like regions in \cite[Def.~4.1.1]{2509.23734}, although there are small technical differences between the two definitions. 
Note that \cite{2509.23734} prove Haag duality for their definition of cone-like regions \cite[Thm.~4.1.4]{2509.23734}.
\end{rem}

\subsection{Canonical state and boundary algebras}

Given a net $(\fA,p)$ satisfying \ref{LTO:QECC}, it was shown in \cite[\S2.3]{MR4945955} that whenever $R \ll_s S$, the scalar $\psi(x)$ defined by 
\[
\psi(x)p_{S} \coloneqq p_{S} x p_{S}
\qquad\qquad\qquad\qquad
x\in \fA(R)
\]
is independent of the choice of $S$, and thus defines a canonical state $\psi$ on the quasilocal algebra $\fA$.
For a region $R$, we define the von Neumann algebra $\cA(R)\coloneqq\fA(R)''$ on $L^2(\fA,\psi)$.
By definition, $\psi$ satisfies $\psi(p_R)=1$ for every region $R$.

\begin{rem}
Suppose $\phi$ is a normal state on a von Neumann algebra $M$.
Recall that the left kernel of $\phi$
$$
N_\phi \coloneqq\set{x\in M}{\phi(x^*x)=0}
$$
is a $\sigma$-WOT closed left ideal of $M$, and thus $N_\phi = M[\phi]^\perp$  for a unique projection $[\phi]^\perp\in M$.
The projection $[\phi]=[\phi]^{\perp\perp}$ is called the \emph{support} of $\phi$;
$\phi$ is faithful on $[\phi] M[\phi]$, and
$$
\phi(x) = \phi([\phi] x[\phi])
\qquad\qquad\qquad
\forall\, x\in M
$$
(see e.g. \cite[Lem.~III.3.6]{MR1873025}).
Observe that for another projection $p\in M$, 
\begin{enumerate}[label=\textup{($[\phi]$\arabic*)}]
\item 
\label{supp:phi(p)=1}
$[\phi]\leq p$
if and only if
$\phi(p)=1$, and
\item 
\label{supp:phiFaithful}
$p\leq [\phi]$
if and only if
$\phi$ is faithful on $pMp$. 
\end{enumerate}
\end{rem}

Consider now the GNS normal unital $*$-algebra homomorphism $\pi_\phi: M\to B(L^2(M,\phi))$.
Since $\ker(\pi_\phi)$ is a $\sigma$-WOT closed 2-sided ideal of $M$, it is of the form $M z_\phi^\perp$ for some cental projection $z_\phi^\perp\in Z(M)$.
Thus, the action of $M$ on $L^2(M,\phi)$ factors through the faithful action of $Mz_\phi$ on $L^2(M,\phi)$.
The following lemma is well-known; we include a short proof for completeness and convenience.

\begin{lem}
\label{lem:KernelOfActionCentralSupportOfState}
We have that $z_\phi = z([\phi])$, where $z([\phi])$ is the central support projection of $[\phi]$.
\end{lem}
\begin{proof}
Since
$
\phi(z_\phi^\perp) = \|z_\phi^\perp\Omega_\phi\|_2^2 = 0,
$
$z_\phi^\perp\in N_\phi=M[\phi]^\perp$, and thus $z_\phi^\perp\leq [\phi]^\perp$.
Conversely, for all $m\in M$, $z([\phi])^\perp m[\phi]=0$.
This immediately implies that 
$$
z([\phi])^\perp m\Omega_\phi
=
z([\phi])^\perp m [\phi]\Omega_\phi 
= 
0
\qquad\qquad\qquad
\forall\, m\in M,
$$
and thus $\pi_\phi(z([\phi])^\perp)=0$.
Hence $z([\phi])^\perp\in \ker(\pi_\phi)=Mz_\phi^\perp$, so $z([\phi])^\perp \leq z_\phi^\perp$.
\end{proof}

\begin{rem}
\label{rem:CentralSupportOne}
Clearly \ref{supp:phi(p)=1} implies that 
$z_\phi\leq z(p)$
and thus 
$\pi_\phi(p)$ has central support 1 on $L^2(M,\phi)$.
\end{rem}

\begin{ex}
\label{ex:CentralSupportAndL2}
Suppose $M\subset B(H)$ is a von Neumann subalgebra and $\Omega\in H$ is a unit vector, and set $K\coloneqq \overline{M\Omega}\cong L^2(M,\phi)$ for $\phi=\omega_\Omega=\langle \Omega|\cdot|\Omega\rangle$.
Let $p_K\in M'$ be the orthogonal projection onto $K$, and let $z(p_K)\in Z(M)$ denote its central support.
We claim that $z_\phi=z(p_K)$.
First, since $p_K\Omega=\Omega$, we have $z(p_K)\Omega=\Omega$, so $\phi(z(p_K))=\langle \Omega|\Omega\rangle =1$, and $z_\phi\leq z(p_K)$.
Second, $Mz(p_K)$ acts faithfully on $K=p_KH$, so $p_K\leq z_\phi$.
\end{ex}

\begin{warn}
Even though $\psi(p_R)=1$ for all bounded regions $R$, it does not follow from \ref{supp:phi(p)=1} and Remark \ref{rem:CentralSupportOne} that $p_R$ automatically has central support one on $\cA(R)$, but rather it always has central support one in $\fA''$ acting on $L^2(\fA,\psi)$.\footnote{A simple counterexample to keep in mind is
the projection $p=|0\rangle\langle 0|\in \bbC\oplus \bbC\subset M_2(\bbC)$ and the vector state $\omega_{|0\rangle}$;
$p$ has central support one in $M_2(\bbC)$ acting on $\bbC^2=L^2(M_2(\bbC),\omega_{|0\rangle})$, but not in $\bbC\oplus \bbC$.
It does, however, have central support one in the image of $\bbC\oplus \bbC$ acting on $L^2(\bbC\oplus \bbC, \omega_{|0\rangle})\cong \bbC$, which only sees the first factor.}
\end{warn}

The axioms \ref{LTO:Boundary}--\ref{LTO:Injective} are used to construct a net of \emph{boundary algebras} on $\partial \Lambda$ for a cone-like region $\Lambda$.
That is, we have a quasilocal algebra $\fB(\Lambda)$ and an assignment of $\rmC^*$-subalgebras $\fB(I)$ to every interval $I\subset \partial\Lambda$ that satisfies the axioms \ref{it:netempty}--\ref{it:netdense}.
Just as \ref{LTO:QECC} gives a canonical state, a similar argument in \cite[\S2.4]{MR4945955} shows that whenever $R \Subset_s S\subset \Lambda$ with $I=\partial R\cap \partial S\subset \partial \Lambda$, the operator $\bbE(x)$ defined by
\[
\bbE(x)p_{S} \coloneqq p_{S} x p_{S}
\qquad\qquad\qquad\qquad
x\in \fA(R)
\]
only depends on sites near $I$ and not on the specific $R \Subset_s S$.
Taking $\fB(I)$ to be the collection of these $\bbE(x)$ defines our desired boundary net.

\begin{exs}
\label{exs:KnownBoundaryAlgebras}
We briefly summarize the description of boundary algebras from topologically ordered commuting projector models known to satisfy \ref{LTO:QECC}--\ref{LTO:Injective} from Examples \ref{exs:KnownTopologicalCommutingProjectorModels}.
\begin{itemize}
\item 
In \cite[\S3]{MR4945955}, the boundary algebra for the toric code was shown to be generated by the `remainders' of terms in the Hamiltonian that straddle the boundary $I$.
See Example~\ref{ex:ToricCodeHD} below for more details.

\item 
In \cite{MR4814524}, the boundary algebra for Kitaev's quantum double model for a finite group $G$ was again shown to be generated by the `remainders' of terms in the Hamiltonian that straddle the boundary $I$.
See Example \ref{ex:QuantumDoubleHD} below for more details.
It was also shown that the canonical state $\psi$ is the normalized trace on the quantum spin system restricted to the boundary algebra (cf.~\cite{MR4721705}).
This analysis is also valid for the toric code which corresponds to $G = \mathbb{Z}_2$.

\item 
In \cite[\S4]{MR4945955}, the boundary algebra for the Levin-Wen model based on a unitary fusion category $\cX$ was shown to be a fusion categorical net.
Setting $X\coloneqq\bigoplus_{x\in\Irr(\cX)} x$, an interval with $n$ sites along $\cK$ gives the local boundary algebra
$$
\fB_n\coloneqq\End_\cX(X^{\otimes n}).
$$
The canonical state $\psi$ restricted to $\fB_n$ was shown to be given by the formula
\[
\psi_n(\varphi)\coloneqq
\frac{1}{D_\cX^n} \sum_{\vec{x}\in \Irr(\cX)^n} d_{\vec{x}} \tr_\cX(p_{\vec{x}}\cdot \varphi),
\]
where we write $\vec{x}=(x_1,\dots, x_n)\in \Irr(\cX)^n$.

\item 
In \cite{2506.19969}, a similar analysis of the 2D boundary net of algebras was given for the Walker--Wang model associated to a unitary braided fusion category $\cA$.
We have a description of this boundary net as the \emph{braided} fusion spin system associated to $\cA$ with distinguished object $A\coloneqq\bigoplus_{a\in\Irr(\cA)}a$.
We get a similar formula for the canonical state as above.

\end{itemize}
One can compute the more general versions of the boundary algebras for these disk-like regions.
As expected, for the toric code and quantum double, the generators of these algebras are the remainders of the terms in the Hamiltonian that straddle $I=\partial R \cap \partial S$,
and for Levin--Wen and Walker--Wang, we get gluing operators that glue along an interval of the 1-sphere-like boundary.
\end{exs}

\begin{rem}
A precise way of formulating the `remainders' of terms in the Hamiltonian is to use interaction algebras \cite{MR1816609}.
One way to formulate our results for the examples above is that the boundary algebras from the LTO axioms are these interaction algebras.
The interaction algebras also were shown to emerge from Osterwalder--Schrader reconstruction for reflection positive quantum spin systems \cite{2510.20662}. 
Further discussion of this point is provided in \textsection\ref{sec:LTOvsRP}.
\end{rem}

\subsection{The net of von Neumann algebras associated to an LTO}
\label{sec:NetsOfVNA}

We now further study the net of von Neumann algebras on the poset of cone-like regions in $\bbR^n$ associated to our LTO $(\fA,p)$.
Recall that $\cA(\Lambda)\coloneqq\fA(\Lambda)''$ acting on $L^2(\fA,\psi)$.

More abstractly, a net of von Neumann algebras on a given poset may satisfy various additional properties.

\begin{defn}
Suppose $\Lambda\mapsto \cM(\Lambda)$ is a net of von Neumann algebras in $B(H)$ assigned to a poset of subsets of a metric space equipped with an order-reversing involution $\Lambda \mapsto \Lambda'$.
We define the following properties that $\cM$ might satisfy.
\begin{enumerate}[label=\textup{(M\arabic*)}]
\item 
\label{M:Isotony}
(Isotony)
$\cM(\Lambda)\subset \cM(\Delta)$ whenever $\Lambda\leq \Delta$.
\item 
\label{M:Locality}
(Locality)
$\cM(\Lambda)\subset \cM(\Delta)'$ whenever $\Lambda\leq \Delta'$.
\item
\label{M:HD}
(Haag duality) 
$\cM(\Lambda)' =  \cM(\Lambda')$ for all cones $\Lambda$.
\item 
\label{M:DisjointAdd}
(Disjoint additivity)
Whenever $\Lambda,\Delta$ are `sufficiently separated,'
$\cM(\Lambda)\vee \cM(\Delta) = \cM(\Lambda \cup \Delta)$.
\item 
\label{M:Coarse}
(Coarse\footnote{The term `coarse' here means an action of the algebraic tensor product extends to the von Neumann spatial tensor product; see \cite{correspondences} or \cite[Def.~2.18]{MR3663592}.})
Whenever $\Lambda,\Delta$ are `sufficiently separated,'
$\cM(\Lambda)\vee\cM(\Delta) \cong \cM(\Lambda)\otimes \cM(\Delta)$.
\item 
\label{M:ApproxSplit}
(Approximate split)
Whenever $\Lambda$ is `sufficiently contained' in $\Delta$, 
there is a type $\rm I$ factor $\cM(\Lambda)\subset \cN\subset \cM(\Delta)$.
\end{enumerate}
\end{defn}

\begin{exs}
Clearly our net of algebras $\Lambda\mapsto \cA(\Lambda)$ from our LTO $(\fA,p)$ in the ground state $\psi$ satisfies \ref{M:Isotony} and \ref{M:Locality}.
The following models were shown to satisfy the following properties in addition:
\begin{itemize}
\item 
1D fusion spin chains satisfy \ref{M:HD} and \ref{M:Coarse} \cite{MR4814692}.
\item The Toric Code satisfies \ref{M:Coarse} \cite{MR2804555} and \ref{M:HD} and \ref{M:ApproxSplit} \cite{MR2956822}.
Note that the proof of \ref{M:ApproxSplit}  in \cite{MR2956822} is constructive.
\item The Quantum Double Model for finite groups $G$ satisfies \ref{M:Coarse} \cite{MR3426207} and \ref{M:HD} and \ref{M:ApproxSplit} \cite{MR3426207, 2509.23734}.
In this case, the proof of \ref{M:ApproxSplit} factors through \ref{D:Factors} below.
\item Levin-Wen models satisfy \ref{M:HD} \cite{2509.23734} and \ref{M:Coarse} and \ref{M:ApproxSplit} \cite{2511.08382}.
\item Braided categorical nets satisfy \ref{M:HD} \cite{2506.19969}.
\end{itemize}
Additionally, \ref{M:DisjointAdd} holds for any quantum spin system, since in this case 
\[
\cM(\Lambda) \vee \cM(\Delta) = \bigvee_{j \in \Lambda \cup \Delta} \fA(\{j\}) = \cM(\Lambda \cup \Delta).
\]
\end{exs}

\begin{rem}
The precise definition of 
`sufficiently separated' in~\ref{M:DisjointAdd} 
and 
`sufficiently contained' in~\ref{M:ApproxSplit} 
depends on the model under consideration.
For concrete 2D models, it is given in terms of a minimum distance between the boundaries of the cones.
For the Levin--Wen model or the Toric Code, this distance can be found explicitly~\cite{MR2956822,2511.08382}.
For more general models in the same phase, one can then use locality bounds for the automorphism connecting the two states, see for example~\cite{MR4426734}.
\end{rem}

\begin{rem}
As we focus on nets of algebras associated to connected regions in this article, we will not explicitly study \ref{M:DisjointAdd} here.
However, it is an interesting condition for fusion spin chains and anyon spin chains \cite{MR4814692} that is not always satisfied.
We refer the reader to \cite{2503.20863} for various conjectures about when \ref{M:HD} and \ref{M:DisjointAdd} hold for such nets of algebras. 
These conjectures are the subject of the upcoming article \cite{FusionNetsOnS1}.
\end{rem}

\begin{facts}
We state some basic dependencies amongst \ref{M:Isotony}--\ref{M:ApproxSplit} cf.~\cite{MR703083}.
\begin{enumerate}[label=(D\arabic*)]
\item 
\ref{M:Isotony} and \ref{M:HD} imply \ref{M:Locality}.

\item 
\label{D:ProperlyInfiniteCone}
Suppose $\cM$ satisfies \ref{M:Isotony}, \ref{M:HD}, and \ref{M:Coarse}.
If $\Lambda, \Delta$ are sufficiently separated, $\Gamma\leq \Lambda'\cap \Delta'$, and $\cM(\Gamma)$ is properly infinite, then there is a type $\rm I$ factor $\cN$ such that $\cM(\Lambda)\subset \cN\subset \cM(\Delta')$.
\begin{proof}
Observe that 
$$
\cM(\Gamma)
\underset{\text{\ref{M:Isotony}}}{\subseteq} 
\cM(\Lambda')\cap \cM(\Delta') 
\underset{\text{\ref{M:HD}}}{=}
\cM(\Lambda)'\cap \cM(\Delta)'
$$
is also properly infinite.
Now \ref{M:Coarse} gives the hypotheses to apply \cite[Cor.~1(ii)]{MR703083}, which immediately implies the result. 
\end{proof}

\item 
\ref{M:ApproxSplit}
implies
$\cM(\Lambda) \vee \cM(\Delta)' \cong \cM(\Lambda) \otimes \cM(\Delta)'$
by
\cite[p368]{MR703083}.
Using~\ref{M:Locality} one gets~\ref{M:Coarse} for $\widehat{\Delta} \leq \Delta'$, or alternatively, it follows for $\Delta'$ using \ref{M:HD}.

\item 
\label{D:Factors}
If every $\cM(\Lambda)$ is a factor, then \ref{M:HD} and \ref{M:Coarse} together imply \ref{M:ApproxSplit}. 
\begin{proof}
In this case, \ref{M:Coarse} and \cite[Cor.~1(iv)]{MR703083} imply that there exists a type I factor $\cN$ such that $\cM(\Lambda) \subset \cN \subset \cM(\Delta)' = \cM(\Delta')$.
\end{proof}
\end{enumerate}

\end{facts}

\begin{rem}
\label{rem:ASBSHD}
Much can still be said in the case where bounded spread Haag duality in the sense of~\cite{MR4927814} is used instead of Haag duality. 
In this case, dependency \ref{D:ProperlyInfiniteCone} still holds, although more separation between $\Lambda$ and $\Delta$ may be required.
In the case where the algebras $\cM(\Lambda)$ are factors, finer distinctions can be made. 
Replacing Haag duality with bounded spread Haag duality still, together with \ref{M:Coarse}, implies \ref{M:ApproxSplit}. 
However, in this case, more space will be necessary for $\Lambda$ to be `sufficiently contained' in $\Delta$ to account for the spread in the bounded spread Haag duality. 
On the other hand, if the slightly weaker notion of bounded spread Haag duality used in \cite{2511.08382} is assumed, then the approximate split property must be weakened to assume that $\Delta$ is wider than $\Lambda$ by an arbitrarily small amount. 
This weaker form of the approximate split property is the one used in \cite{MR4426734} and is stable, in the sense that all states in the same phase satisfy it.
\end{rem}

We now include some basic results about the \emph{approximate split property} \ref{M:ApproxSplit}.
We define $\omega \coloneqq \omega_{\Omega_\psi}=\langle \Omega_\psi| \cdot|\Omega_\psi\rangle$ on $B(L^2(\fA,\psi))$.

\begin{lem}
\label{lem:SufficientlySeparatedState}
Suppose $(\fA,p)$ is an LTO with canonical state $\psi$. 
Whenever $\Lambda,\Delta$ are sufficiently separated, 
\[
\omega(xy) = \omega(x)\omega(y)
\qquad\qquad\qquad\qquad
\forall\, x\in \cA(\Lambda),\, y\in \cA(\Delta).
\]
\end{lem}
\begin{proof}
We first check that $\psi(xy) = \psi(x) \psi(y)$ for $x \in \fA(R)$ and $y \in \fA(S)$, where $R \subset \Lambda$ and $S \subset \Delta$ are disk-like.
Let $\widehat{R}$ be disjoint from $S$ such that $R \ll_s \widehat{R}$. 
(Such an $\widehat{R}$ exists because $R \subset \Lambda$, $S \subset \Delta$, and $\Lambda$ and $\Delta$ are sufficiently separated.)
Let $\widehat{S}$ be such that $R, S \ll_s \widehat{S}$. 
For all $x \in \fA(R)$ and $y \in \fA(S)$, 
\[
\psi(x y) p_{\widehat{S}}
=
p_{\widehat{S}} xy p_{\widehat{S}}
=
p_{\widehat{S}} p_{\widehat{R}} x p_{\widehat{R}} y p_{\widehat{S}}
=
\psi(x) p_{\widehat{S}} p_{\widehat{R}} y p_{\widehat{S}}
=
\psi(x) \psi(y) p_{\widehat{S}},
\]
so $\psi(xy) = \psi(x)\psi(y)$. 
Now, $\cA(\Lambda) = (\bigcup_{R \subset \Lambda} \fA(R))''$ and $\cA(\Delta) = (\bigcup_{S \subset \Delta} \fA(S))''$.
The result thus follows by normality of $\omega=\omega_{\Omega_\psi}$, the Kaplansky Density Theorem, and the fact that multiplication is jointly SOT-continuous on bounded sets.
\end{proof}

The following lemma is well-known; we include a short proof for completeness and convenience of the reader.

\begin{lem}[{\cite[Rem.~p365]{MR703083},~c.f.~\cite{MR100798}}]
\label{lem:SeparatedTensorProductState}
Suppose $M$ is a von Neumann algebra with normal state $\phi$.
If $A,B\subset M$ are commuting von Neumann subalgebras such that
\begin{itemize}
\item $\phi(ab)=\phi(a)\phi(b)$ for all $a\in A$ and $b\in B$, and
\item $\phi$ has central support one on $A\vee B\subset M$,\footnote{Even if $\phi$ has central support one on both $A$ and $B$, it does not necessarily have central support one on $A\vee B$.
A simple counterexample to keep in mind is
$$M=\bbC^3,
\qquad
A=
\operatorname{span}\set{
\begin{bmatrix}
a
\\
b
\\
b
\end{bmatrix}
}{a,b\in\bbC},
\qquad
B=
\operatorname{span}\set{
\begin{bmatrix}
a
\\
a
\\
b
\end{bmatrix}
}{a,b\in\bbC},
\qquad\text{and}\qquad
\phi=\frac{1}{2}
\left\langle
\begin{bmatrix}
1
\\
0
\\
1
\end{bmatrix}
\middle|
\,\,\cdot\,\,
\middle|
\begin{bmatrix}
1
\\
0
\\
1
\end{bmatrix}
\right\rangle.
$$
}
\end{itemize}
then $A\vee B\cong A\otimes B$.
\end{lem}
\begin{proof}
We may assume that $M=A\vee B$.
Since $\phi$ has central support one on $M$, $M$ acts faithfully on $L^2(M,\phi)$ by Lemma \ref{lem:KernelOfActionCentralSupportOfState}.
We also claim that $\phi$ has central support one on both $A$ and $B$.
Indeed, let $z_A\in Z(A)$ and $z_B\in Z(B)$ be the central supports of $\phi|_A$ and $\phi|_B$.
Then since $z_Az_B\in P(Z(M))$ and
\[
\phi(z_Az_B)=\phi(z_A)\phi(z_B) = 1,
\]
we have $1=z_\phi\leq z_Az_B$ by Lemma \ref{lem:KernelOfActionCentralSupportOfState}, which implies $z_A=1_A$ and $z_B=1_B$.

Hence we have a spatial isomorphism
\[
L^2(A,\phi)\otimes L^2(B,\phi)
\xrightarrow{ a \Omega_A\otimes  b \Omega_B \mapsto ab\Omega_M}
L^2(M,\phi)
\]
that clearly intertwines the $A\otimes 1$ and $1\otimes B$ actions on $L^2(A,\phi)\otimes L^2(B,\phi)$
with the $A$ and $B$ actions on $L^2(M,\phi)$.
The result follows.
\end{proof}

\begin{cor}
Whenever $\omega$ has central support one on $\cA(\Lambda)\vee\cA(\Delta)$ for sufficiently separated $\Lambda,\Delta$, then 
$\cA(\Lambda)\vee\cA(\Delta)\cong \cA(\Lambda)\otimes \cA(\Delta)$, i.e., \ref{M:Coarse} holds.
\end{cor}
\begin{proof}
Immediate from Lemmas \ref{lem:SufficientlySeparatedState} and \ref{lem:SeparatedTensorProductState}.
\end{proof}

\begin{rem}
For an actual quantum spin system (that is, when $\fA = \bigotimes_{v \in \bbZ^n} M_{n_v}$), $\cA(\Lambda) \vee \cA(\Delta)$ is necessarily a factor, so $\omega$ must have central support one. 
In this case, assuming Haag duality \ref{M:HD},  \ref{M:Coarse} and \ref{M:ApproxSplit} hold by \ref{D:Factors}, with a slightly weaker result in the case of bounded spread Haag duality (see Remark \ref{rem:ASBSHD}).
The proof for the spin system case is the subject of \cite[Thm.~3.29]{2511.08382}.
\end{rem}

\subsection{The boundary algebra as a compression}

For each cone-like region $\Lambda$, we define $p_\Lambda\coloneqq  \bigwedge_{R\subset \Lambda} p_R$, the infimum over all  $p_R$ associated to disk-like regions $R\subset \Lambda$.
Note that $p_\Lambda$ is the $\sigma$-WOT/SOT limit of the projections $p_R$ as the $R$ increase to $\Lambda$.
Since $\omega=\langle \Omega_\psi| \cdot |\Omega_\psi\rangle$,  $\omega(p_\Lambda)=1$ for every cone $\Lambda$.

\begin{lem}
\label{lem:BoundaryCommutesWithConeProjection}
$p_\Lambda$ commutes with $\fB(\Lambda)$ on $L^2(\fA,\psi)$.
\end{lem}
\begin{proof}
Suppose $x\in \fB(I)$ where $I=\partial R\cap \partial S$ for $R\Subset S\subset \Lambda$.
Then for every $S\subset T\subset \Lambda$, it follows that
$I=\partial R\cap \partial T$, and thus $[x,p_{T}]=0$ by definition of $\fB(I)$.
Since $p_T\searrow p_\Lambda$ as $T$ increases to $\Lambda$, 
\[
p_\Lambda x = \lim{}^{SOT}\, p_Tx = \lim{}^{SOT}\, xp_T = xp_\Lambda.
\]
Since $\fB(\Lambda)=\varinjlim \fB(I)$, the result follows.
\end{proof}

\begin{prop}
\label{prop:CornerOfC*Boundary}
$\overline{\fB(\Lambda)p_\Lambda}^{\| \cdot \|} =\overline{p_\Lambda \fA(\Lambda)p_\Lambda}^{\| \cdot \|}$
as operators on $L^2(\fA,\psi)$.
\end{prop}
\begin{proof}
Since $p_\Lambda$ and $\fB(\Lambda)$ commute by Lemma \ref{lem:BoundaryCommutesWithConeProjection},  $\fB(\Lambda)p_\Lambda \subset p_\Lambda \fA(\Lambda)p_\Lambda$, and thus $\overline{\fB(\Lambda)p_\Lambda} \subseteq \overline{p_\Lambda \fA(\Lambda)p_\Lambda}$.
Conversely, for each disk-like region $R\subset \Lambda$, we can find a disk-like region $S\subset\Lambda$ with  $R\Subset S$.
So for every $x\in \fA(R)$, 
$$
p_\Lambda xp_\Lambda = p_\Lambda p_Sxp_Sp_\Lambda \in \fB(R\Subset S)p_\Lambda \subset \fB(\Lambda)p_\Lambda.
$$
Then
$$
p_\Lambda\fA(\Lambda)p_\Lambda
=
p_\Lambda\overline{\bigcup_{R\subset \Lambda} \fA(R)}p_\Lambda 
\subseteq 
\overline{\bigcup_{R\subset \Lambda} p_\Lambda\fA(R)p_\Lambda}
\subseteq 
\overline{\fB(\Lambda)p_\Lambda},
$$ 
so $\overline{p_\Lambda\fA(\Lambda)p_\Lambda}\subseteq \overline{\fB(\Lambda)p_\Lambda}$. 
\end{proof}

\begin{cor}
\label{cor:pLambdaInjective}
If $\psi$ is faithful on $\fB(\Lambda)$, then the multiplication map $(\,\cdot\,) p_\Lambda : \fB(\Lambda)\to \fB(\Lambda)p_\Lambda$ is injective and $\fB(\Lambda)p_\Lambda =p_\Lambda \fA(\Lambda)p_\Lambda$ are both closed.
\end{cor}
\begin{proof}
Suppose $I\subset \partial \Lambda$ is an interval and $R\Subset S\subset \Lambda$ with $I=\partial R\cap \partial S$.
We claim that multiplication by $p_\Lambda$ on $\fB(R\Subset S)\to \fB(R\Subset S)p_\Lambda$ is injective.
Indeed,
if $x p_\Lambda = 0$, then 
$$
\omega(p_\Lambda x^* x p_\Lambda) = \omega(x^*x) = \psi(x^*x) = 0,
$$ 
where the first equality follows from \cite[\textsection 2.1.1]{MR2345476} (see \cite[Lem.~2.22]{MR4945955}).
By faithfulness of $\psi$ on $\fB(\Lambda)$, $x^*x = 0$ so $x =0$.

Now since $\fB(\Lambda)=\varinjlim \fB(I)$ and injective maps are isometric, it follows that multiplication by $p_\Lambda$ on $\fB(\Lambda)$ is injective and $\fB(\Lambda)p_\Lambda$ is closed.
It follows immediately that
$$
\fB(\Lambda)p_\Lambda 
\subseteq 
p_\Lambda \fA(\Lambda)p_\Lambda
\subseteq 
\overline{\bigcup_{R\subset \Lambda} p_\Lambda\fA(R)p_\Lambda}
\subseteq 
\fB(\Lambda)p_\Lambda,
$$
so all these inclusions are equalities.
\end{proof}

\begin{rem}
The article \cite[\textsection 6]{2507.03201} introduces a strengthened \ref{LTO:Injective} axiom in order to make a connection between LTO boundary algebras and their frustration-free Hamiltonian setup.
We can state a version of their strengthened axiom for abstract quantum spin systems as follows:
\begin{enumerate}[label=\textup{(LTO4$^{+}$)}]
\item
\label{LTO:SoupedUpInjective}
For a cone-like region $\Lambda$, let $p^{**}_\Lambda$ be the limit of the projections $p_R$ associated to disk-like regions $R\subset \Lambda$ as they increase to $\Lambda$, taken  $\sigma$-WOT in the universal enveloping von Neumann algebra $\fA(\Lambda)^{**}$.
Whenever $R\Subset S\subset \Lambda$ with $\partial R\cap \partial S \subset \partial \Lambda$,  
the map $x\mapsto xp^{**}_\Lambda$ is injective on $\fB(R \Subset S)$.
\end{enumerate}
In \cite[Thm.~6.24]{2507.03201}, it is proven for a quantum spin system that assuming \ref{LTO:SoupedUpInjective},
$$
\overline{p_{\bbH}^{**} \fA(\bbH)p_{\bbH}^{**}}^{\| \cdot \|} = \overline{\fB p_{\bbH}^{**}}^{\| \cdot \|}.
$$ 
They further conjecture that our axioms \ref{LTO:QECC}--\ref{LTO:Injective} are sufficient for their result \cite[Conj.~6.20]{2507.03201}.

Now clearly \ref{LTO:SoupedUpInjective} implies \ref{LTO:Injective}.
The converse is true provided $\psi$ is faithful on $\fB(\Lambda)$;
one can even show that multiplication by $p_\Lambda^{**}$ is injective on $\fB(\Lambda)$.
Indeed, by the Sherman--Takeda Theorem,\footnote{Suppose $A$ is a $\rmC^*$-algebra and $\varphi$ is a state on $A$.
By the Sherman--Takeda Theorem \cite[Thm.~III.2.4]{MR1873025}, we may identify the universal enveloping von Neumann algebra $A^{**}$ with $A''$ in its universal representation:
\[
\pi_u :A \to B\left(\bigoplus_{\phi \in S(A)} L^2(A,\phi)\right).
\]
Setting $\rmW^*_\varphi(A)\coloneqq A''$ acting on $L^2(A,\varphi)$,
by universality (see \cite[Defn.~III.2.3]{MR1873025}), there is a surjective $\sigma$-WOT continuous map $A^{**}\to \rmW^*_\varphi(A)$.
Thus there is a central projection $z\in Z(A^{**})$ such that $\rmW^*_\varphi(A) \cong zA^{**}$.}
we may identify $\cA(\Lambda) = z\fA(\Lambda)^{**}$
for some central projection $z\in Z(\fA(\Lambda)^{**})$.
Since $p_R\searrow p^{**}_\Lambda$ in $\fA(\Lambda)^{**}$ and $p_Rz\searrow p_\Lambda$ in $\cA(\Lambda) = \fA(\Lambda)^{**}z$, we have $p_\Lambda = p_\Lambda^{**}z$.
So if $x\in \fB(\Lambda)$ such that $xp_\Lambda^{**}=0$, then $xp_\Lambda=xp_\Lambda^{**}z=0$, 
and the result follows by Corollary \ref{cor:pLambdaInjective}.

\end{rem}

The following two properties that an LTO $(\fA,p)$ might satisfy in the canonical state $\psi$ for a cone-like region $\Lambda$ will be of particular importance.
\begin{enumerate}[label=\textup{($z_\psi$\arabic*)}]
\item 
\label{LTO:CentralSupportOne}
$p_\Lambda$ has central support one in $\cA(\Lambda)\coloneqq\fA(\Lambda)''$ acting on $L^2(\fA,\psi)$.
\item 
\label{LTO:FaithfulOnBoundary}
the canonical state $\omega$ is faithful on the compression $\cB(\Lambda)\coloneqq  p_\Lambda \cA(\Lambda)p_\Lambda$. 
\end{enumerate}
One should compare \ref{LTO:CentralSupportOne} with \ref{supp:phi(p)=1} and \ref{LTO:FaithfulOnBoundary} with \ref{supp:phiFaithful}.

\begin{exs}
\label{exs:TopOrderModelsSatisfyExtraAxioms}
    Clearly \ref{LTO:CentralSupportOne} holds whenever each $\cA(\Lambda)$ is a factor, e.g., when $(\fA,p)$ comes from an honest quantum spin system as in Examples \ref{exs:KnownTopologicalCommutingProjectorModels}.\footnote{That $\cA(\Lambda)$ is a factor in this case follows because the canonical state $\psi$ is pure and because disjoint additivity~\ref{M:DisjointAdd} holds.}
Moreover, each of these examples is also known to satisfy \ref{LTO:FaithfulOnBoundary} for every cone-like region $\Lambda$.
\end{exs}

\begin{rem}
Since $\omega(p_\Lambda)=1$, 
we always have
$[\omega|_{\cA(\Lambda)}]\leq p_\Lambda$ by \ref{supp:phi(p)=1}.
Hence \ref{LTO:FaithfulOnBoundary} is equivalent to 
the support $[\omega|_{\cA(\Lambda)}]=p_\Lambda$ by \ref{supp:phiFaithful}.
\end{rem}

\begin{question}
Do \ref{LTO:CentralSupportOne} or \ref{LTO:FaithfulOnBoundary} automatically hold for all LTOs?
\end{question}

We now give some applications of the conditions from \ref{LTO:CentralSupportOne} and/or \ref{LTO:FaithfulOnBoundary}
to the algebras associated to a cone-like region $\Lambda$.

\begin{lem}
\label{lem:Quasiequivalent}
Assuming
\ref{LTO:CentralSupportOne} and \ref{LTO:FaithfulOnBoundary},  
$\cA(\Lambda)\cong \fA(\Lambda)''$ acting on $L^2(\fA(\Lambda),\psi)$, i.e., the representations $L^2(\fA,\psi)$ and $L^2(\fA(\Lambda),\psi)$ of $\fA(\Lambda)$ are quasi-equivalent.
In particular, $\cA(\Lambda)$ acts faithfully on $L^2(\cA(\Lambda),\omega)$. 
\end{lem}
\begin{proof}
Observe that we may identify $L^2(\cA(\Lambda),\omega)=\overline{\cA(\Lambda)\Omega_\psi}$
and
$L^2(\fA(\Lambda),\psi)=\overline{\fA(\Lambda)\Omega_\psi}$
as subspaces of $L^2(\fA,\psi)$.
Since $\fA(\Lambda)$ is SOT-dense in $\cA(\Lambda)$ acting on $L^2(\fA,\psi)$,
we see that $\cA(\Lambda)\Omega_\psi \subset \overline{\fA(\Lambda)\Omega_\psi}$,
and thus we may identify
$L^2(\cA(\Lambda),\omega)=L^2(\fA(\Lambda),\psi)$.

It remains to prove that $\cA(\Lambda)$ acts faithfully on $L^2(\cA(\Lambda),\omega)$.
(Indeed, since $z_\omega \in \cA(\Lambda)$ is the central support of the projection onto $L^2(\cA(\Lambda),\omega)\subset L^2(\fA,\psi)$ as discussed in Example \ref{ex:CentralSupportAndL2}, quasi-equivalence follows by \cite[Thm.~10.3.3(ii)]{MR1468230}.)
If $z\in Z(\cA(\Lambda))$ is a central projection that acts as zero on $L^2(\cA(\Lambda),\omega)$,
then $zp_\Lambda \Omega_\psi=0$, which implies $zp_\Lambda=0$ as $\Omega_\psi$ is separating for $\cB(\Lambda)$ by \ref{LTO:FaithfulOnBoundary}.
But $zp_\Lambda=0$ implies $z=z\cdot z(p_\Lambda)=0$ by \ref{LTO:CentralSupportOne}.
\end{proof}

The next proposition shows that assuming \ref{LTO:CentralSupportOne} and \ref{LTO:FaithfulOnBoundary}, $\cB(\Lambda)$ can be identified with $\fB(\Lambda)''$ acting on $L^2(\fB(\Lambda),\psi)$, where $\fB(\Lambda)=\varinjlim \fB(I)$ is the boundary algebra along $\partial \Lambda$.
This subtlety was overlooked in \cite[\S5.4]{MR4945955}; however, the results stated there are still correct as they were only stated for the examples discussed in Examples \ref{exs:KnownTopologicalCommutingProjectorModels} and \ref{exs:TopOrderModelsSatisfyExtraAxioms}.

\begin{prop}
\label{prop:CompressionvNA}
Assuming
\ref{LTO:CentralSupportOne} and \ref{LTO:FaithfulOnBoundary}, 
$\cB(\Lambda)\cong \fB(\Lambda)''$ acting on $L^2(\fB(\Lambda),\psi)$.
\end{prop}
\begin{proof}
By Proposition \ref{prop:CornerOfC*Boundary}, $\overline{\fB(\Lambda)p_\Lambda}^{\|\cdot\|}= \overline{p_\Lambda\fA(\Lambda)p_\Lambda}^{\|\cdot\|}$, which allows us to identify
\begin{align*}
L^2(\fB(\Lambda),\psi) 
&= 
\overline{\fB(\Lambda) \Omega_\psi}^{\|\cdot\|_2}
=
\overline{\fB(\Lambda) p_\Lambda\Omega_\psi}^{\|\cdot\|_2}
=
\overline{p_\Lambda\fA(\Lambda)p_\Lambda\Omega_\psi}^{\|\cdot\|_2}
\\&=
p_\Lambda\overline{\fA(\Lambda)\Omega_\psi}^{\|\cdot\|_2}
=
p_\Lambda L^2(\fA(\Lambda),\psi)
=
p_\Lambda L^2(\cA(\Lambda),\omega).
\end{align*}
We thus see that the action of $\fB(\Lambda)$ on $L^2(\fB(\Lambda),\psi)$ factors through the action of $\overline{\fB(\Lambda)p_\Lambda}$.
Again by Proposition \ref{prop:CornerOfC*Boundary},
we may identify the action of 
$\overline{\fB(\Lambda)p_\Lambda}$ on $L^2(\fB(\Lambda),\psi)$
with the action of
$\overline{p_\Lambda\fA(\Lambda)p_\Lambda}$ on 
$p_\Lambda L^2(\cA(\Lambda),\omega)$,
so both algebras generate isomorphic von Neumann algebras.
The former generates $\fB(\Lambda)''$ by definition, and we claim that the latter generates $\cB(\Lambda)=p_\Lambda \cA(\Lambda)p_\Lambda$.
We prove this latter claim in three steps.

\item[\underline{Step 1:}]
The SOT-closures of 
$p_\Lambda\fA(\Lambda)p_\Lambda$
and
$\overline{p_\Lambda\fA(\Lambda)p_\Lambda}$
on $p_\Lambda L^2(\cA(\Lambda),\omega)$
agree.
Indeed, the SOT-closure of $p_\Lambda\fA(\Lambda)p_\Lambda$ is norm closed and thus contains $\overline{p_\Lambda\fA(\Lambda)p_\Lambda}$.

\item[\underline{Step 2:}]
The SOT-closures of $p_\Lambda\fA(\Lambda)p_\Lambda$ on $p_\Lambda L^2(\cA(\Lambda),\omega)$
and
$L^2(\cA(\Lambda),\omega)$
agree.
Indeed, if $p_\Lambda a_ip_\Lambda \eta \to x\eta$ for all $\eta\in L^2(\cA(\Lambda),\omega)$, then $x=p_\Lambda xp_\Lambda$, and thus one can check SOT-convergence solely on $p_\Lambda L^2(\cA(\Lambda),\omega)$.

\item[\underline{Step 3:}]
The SOT-closure of $p_\Lambda\fA(\Lambda)p_\Lambda$
on $p_\Lambda L^2(\cA(\Lambda),\omega)$
is $\cB(\Lambda)$.
Indeed, the action of $\fA(\Lambda)$ on $L^2(\cA(\Lambda),\omega)$ 
extends to a faithful normal action of $\cA(\Lambda)$
by Lemma \ref{lem:Quasiequivalent} (this is where we use \ref{LTO:CentralSupportOne} and \ref{LTO:FaithfulOnBoundary}),
and since multiplication by $p_\Lambda$ is SOT-continuous, we see that the SOT-closure of $p_\Lambda\fA(\Lambda)p_\Lambda$ is exactly $p_\Lambda\cA(\Lambda)p_\Lambda$.
\end{proof}

\subsection{Conditional expectations}
\label{sec:ConditionalExpectations}

For this section, we fix a cone-like region $\Lambda$ and write $\fB$ for the boundary algebra along $\partial \Lambda$.
A new feature that arises with the more general disk-like regions is that whenever $I_1\subset I_2$, we get a canonical $\psi$-preserving \emph{conditional expectation} $\fB(I_2)\to \fB(I_1)$ as follows.

\begin{construction}
\label{construct:ConditionalExpectations}
Choose $R_i \Subset_s S$ with 
$R_1\subset R_2$ 
and
$\partial R_i \cap \partial S = I_i$ for $i=1,2$.
(Taking $R_2\supset R_1$ and a large $\Delta$ will not affect the boundary algebras by \ref{LTO:Surjective} and \ref{LTO:Injective}.)
We then pick a more general region $\widehat{S}$ that satisfies $S \subset \widehat{S}$ and $\partial R_2\cap \partial \widehat{S} =  I_1$.
Here is a cartoon of an example.
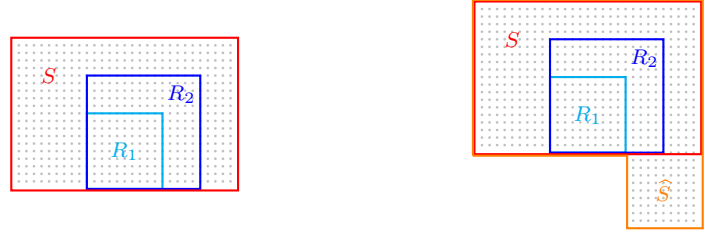
\[
\tikzmath{
\foreach \x in {-9,...,19}{
\foreach \y in {1,...,19}{
\filldraw[gray!50!white] ($ .1*(\x,\y) $) circle (.01cm);
}}
\draw[thick, cyan] (0,0) rectangle (1,1);
\node[cyan] at (.5,.5) {$\scriptstyle R_1$};
\draw[thick, blue] (0,0) rectangle (1.5,1.5);
\node[blue] at (1.25,1.25) {$\scriptstyle R_2$};
\draw[thick, red] (-1,-.02) rectangle (2,2);
\node[red] at (-.5,1.5) {$\scriptstyle S$};
}
\qquad\qquad\qquad\qquad
\tikzmath{
\foreach \x in {-9,...,19}{
\foreach \y in {1,...,19}{
\filldraw[gray!50!white] ($ .1*(\x,\y) $) circle (.01cm);
}}
\foreach \x in {11,...,19}{
\foreach \y in {-9,...,0}{
\filldraw[gray!50!white] ($ .1*(\x,\y) $) circle (.01cm);
}}
\draw[thick, cyan] (0,0) rectangle (1,1);
\node[cyan] at (.5,.5) {$\scriptstyle R_1$};
\draw[thick, blue] (0,0) rectangle (1.5,1.5);
\node[blue] at (1.25,1.25) {$\scriptstyle R_2$};
\draw[thick, red] (-1,-.02) rectangle (2,2);
\node[red] at (-.5,1.5) {$\scriptstyle S$};
\draw[thick, orange] (-1.02,-.04) -- (1.02,-.04) -- (1.02,-1) -- (2.02,-1) -- (2.02,2.02) -- (-1.02,2.02) -- (-1.02,-.04);
\node[orange] at (1.5,-.5) {$\scriptstyle \widehat{S}$};
}
\]
The map
$\fB(I_2)\cong \fB(R_2\Subset_s S) \to \fB(R_2 \Subset_s \widehat{S})$ given by $xp_S \mapsto p_{\widehat{S}} xp_{\widehat{S}}$ is a surjective ucp map onto
\[
\fB(R_2 \Subset_s \widehat{S}) 
\underset{\text{\ref{LTO:Surjective}}}{=}
\fB(R_1 \Subset_s \widehat{S}) 
\underset{\text{\ref{LTO:Injective}}}{\cong}
\fB(R_1 \Subset_s S)
\cong
\fB(I_1)
\]
that is clearly $\fB(R_1 \Subset_s \widehat{S})\cong \fB(I_1)$ bimodular.
We thus have our desired conditional expectation $\fB(I_2)\to \fB(I_1)$.
Taking the inductive limit as $I_2\to \infty$, we get an inductive limit conditional expectation $\fB\to \fB(I_1)$ by \cite[Cor.~2.19]{MR4945955}.
\end{construction}

Since this conditional expectation preserves the canonical state $\psi$, it extends to a normal conditional expectation $\fB''\to \fB(I)''$ in $B(L^2(\fB,\psi))$ in the usual way, under the additional assumption that $\psi$ is faithful on $\cB$, 
cf.~\ref{supp:phiFaithful}.\footnote{Faithfulness on $\fB$ may not always imply faithfulness on $\fB''$; see \cite{MR355622}.}
In more detail, we have the following lemma.

\begin{lem}
Suppose $A\subset B$ is a unital inclusion of $\rmC^*$-algebras, $\phi$ is a state on $B$, and $E: B\to A$ is a $\phi$-preserving conditional expectation.
Let $\Omega_A,\Omega_B$ denote the canonical cyclic vectors in $L^2(A,\phi), L^2(B,\phi)$ respectively, and let $\pi_A:A\to B(L^2(A,\phi))$ and $\pi_B:B\to B(L^2(B,\phi))$ denote the GNS representations.
Assume that the extension $\langle \Omega_B | \,\cdot\,|\Omega_B\rangle$ of $\phi$ to $\pi_B(B)''$, still denoted $\phi$, is faithful.
\begin{enumerate}[label=(\arabic*)]
\item
    \label{it:isometry}
$\iota : L^2(A,\phi)\to L^2(B,\phi)$ given by $a\Omega_A\mapsto a\Omega_B$ is a well-defined isometry.
\item
    \label{it:expconjugate}
For all $b\in B$, $E(b)=\iota^* b\iota \in B(L^2(A,\phi))$, and thus $E$ extends to a normal ucp map $\pi_B(B)''\to \pi_A(A)''$, which we also denote by $E$, satisfying
$\langle \Omega_B| x|\Omega_B\rangle=\langle \Omega_A| E(x)|\Omega_A\rangle$ for all $x\in \pi_B(B)''$.
\item
    \label{it:stariso}
The representation $(\pi_B)|_A: A\to B(L^2(B,\phi))$ extends to a normal representation of $\pi_A(A)''$ that is a $*$-isomorphism onto $\pi_B(A)''$.
\item
    \label{it:normcondexp}
The map $b\mapsto \iota^*b\iota$ is a well-defined normal conditional expectation from $\pi_B(B)''\to \pi_B(A)''$ on $L^2(B,\phi)$.
\end{enumerate}
\end{lem}
\begin{proof}
Statement~\ref{it:isometry} follows immediately from 
\[
\|\iota a\Omega_A\|^2_\phi = \|a\Omega_B\|_\phi^2 = \phi(a^*a) = \|a\Omega_A\|_\phi^2.
\]
Statement~\ref{it:expconjugate} is verified by taking inner products against $a\Omega_A, c\Omega_A$ for $a,c\in A$:
\[
\langle a\Omega_A | \iota^*b \iota c\Omega_A\rangle 
=
\langle a\Omega_B| bc\Omega_B\rangle
=
\phi(a^*bc)
=
\phi(E(a^*bc))
=
\phi(a^*E(b)c)
=
\langle a\Omega_A | E(b)c\Omega_A\rangle.
\]
The formula $x\mapsto \iota^* x\iota$ is manifestly normal and ucp on $\pi_B(B)''$, 
so it takes values in 
$\pi_A(A)''$ by normality and the Kaplansky Density Theorem.

Now for~\ref{it:stariso}, to see that $(\pi_B)_A$ extends to a normal representation of $\pi_A(A)''$, 
it suffices to find a dense subset $D\subset L^2(B,\phi)$ such that for every $\eta_1,\eta_2\in D$,
there are $\xi_1,\xi_2\in L^2(A,\phi)$ such that
\[
\langle \eta_1 | \pi_B(a) \eta_2 \rangle_{L^2(B,\phi)}= \langle \xi_1 | \pi_A(a)\xi_2\rangle_{L^2(A,\phi)}
\qquad\qquad\qquad
\forall\, a\in A.
\] 
Since $\phi$ is faithful on $\pi_B(B)''$, we can consider the corresponding modular automorphism group $\sigma^\phi$.
Let $D=B_0\Omega_B$ where $B_0$ is the Tomita algebra of $\pi_B(B)''$, a $*$-algebra contained in the set of entire elements for $\sigma^\phi$.
Note that $D$ is dense in $L^2(B,\phi)$ \cite[p99-102]{MR1943006}.
The KMS condition implies that for all $b,c\in B_0$,
\begin{align*}
\langle b\Omega_B | \pi_B(a) c\Omega_B\rangle_{L^2(B,\phi)}
&=
\varphi(b^* a c)
=
\varphi(\sigma^\phi_{i}(c)b^* a)
=
\varphi(E(\sigma^\phi_{i}(c)b^*) a)
=
\langle E(b \sigma^\phi_{-i}(c^*)) \Omega_A |  a\Omega_A\rangle_{L^2(A,\phi)},
\end{align*}
as desired, where we used that $E$ can be extended to $\pi_B(B)''$ by~\ref{it:expconjugate}.

Now by normality, we immediately have that the image of $(\pi_B)|_A$ on $\pi_A(A)''$ lands in $\pi_B(A)''$.
To see injectivity, observe that $\pi_A = E \circ (\pi_B)|_A$. 
By normality, this equality holds on the extension to $\pi_A(A)''$.
Thus the extension of $(\pi_B)_A$ to $\pi_A(A)''$ has left inverse $E$ and is thus injective.
To show $(\pi_B)|_A$ is also surjective, it suffices to prove $E$ is injective on $\pi_B(A)''$, since then $E$ has a left-inverse that must be equal to $(\pi_B)|_A$ as well.
Suppose $x\in \pi_B(A)''$ with $E(x)=0$.
By the Kaplansky Density Theorem, there is a bounded sequence $(x_n)\subset A$ with $\pi_B(x_n)\to x$ SOT, and thus $\sigma$-WOT.
Since $E$ is normal, we have
$x_n=E(\pi_B(x_n))\to E(x)=0$ $\sigma$-WOT.
Since $(\pi_B)|_A$ is normal, we see that $x=0$.

The final claim~\ref{it:normcondexp} follows immediately.
\end{proof}

\begin{cor}
\label{cor:ModularGroupPreservesSubalgebras}
Suppose the extension of $\psi$ is faithful on $\fB(\Lambda)''$.\footnote{Note that if we assume \ref{LTO:CentralSupportOne} and \ref{LTO:FaithfulOnBoundary}, we may identify $\fB(\Lambda)''$ with $\cB(\Lambda)$ by Proposition \ref{prop:CompressionvNA}.}
The modular automorphism group $\sigma^\psi$ on $\fB(\Lambda)''$ preserves each of the subalgebras $\fB(I)''$.
\end{cor}
\begin{proof}
When $\psi$ is faithful, the 
$\psi$-preserving conditional expectations $\fB(\Lambda)\to \fB(I)$ from Construction \ref{construct:ConditionalExpectations} extend to normal conditional expectations $\fB''\to \fB(I)''$ on $L^2(\fB(\Lambda),\psi)$.
Now apply Takesaki's Theorem \cite{MR0303307}.
\end{proof}

\begin{rem}
If we think of the modular automorphism group $\sigma^\psi$ as describing some dynamics on the boundary (and $\psi$ a thermal state for this dynamics), the corollary implies that the dynamics have a zero-velocity Lieb--Robinson bound.
This has also been observed in the tensor network picture, where for certain commuting parent Hamiltonians in the bulk, the boundary state is Gibbs state for commuting interactions (cf. Footnote 1 and the paragraph below Main Result 1 in~\cite{MR3927082}).
More generally, locality properties of the boundary state are conjectured to be related to the spectral gap in the bulk~\cite{MR3927082}.
\end{rem}

\begin{ex}
For the Levin-Wen model, if we take intervals $I_{n-1}\subset I_n$ along our cut $\cK$ where $I_n$ is obtained from $I_{n-1}$ by adding a single site, the unique $\psi$-preserving conditional expectation $E^{n}_{n-1} : \fB_{n}\to \fB_{n-1}$ is given by
\[
E^{n}_{n-1}(\varphi) \coloneqq  \frac{1}{D_\cX} \sum_{x_{n}\in\Irr(\cX)} d_{x_{n}} 
\tikzmath{
\draw (-.4,-.3) -- (-.4,-.7) node[below]{$\scriptstyle X$};
\draw (-.4,.3) -- (-.4,1.1) node[above]{$\scriptstyle X$};
\node at (-.05,.7) {$\cdots$};
\node at (-.05,-.5) {$\cdots$};
\draw (.2,-.3) -- (.2,-.7) node[below]{$\scriptstyle X$};
\draw (.2,.3) -- (.2,1.1) node[above]{$\scriptstyle X$};
\draw (.4,.3) -- (.4,.7) arc(180:0:.3cm) --node[right]{$\scriptstyle \overline{x_{n}}$} (1,-.3) arc(0:-180:.3cm);
\roundNbox{}{(0,0)}{.3}{.3}{.3}{$\varphi$}
\filldraw (.4,.5) node[right]{$\scriptstyle x_n$} circle (.05cm);
}
\qquad
\qquad
\text{where}
\qquad
\qquad
\tikzmath{
\draw (0,-.5) -- (0,.5);
\filldraw (0,0) node[right]{$\scriptstyle x_n$} circle (.05cm);
}
=p_{x_{n}}
\in\End_\cX(X).
\]
One may verify by direct calculation that this conditional expectation is obtained by the procedure in Construction \ref{construct:ConditionalExpectations}.
We suppress the proof, as it is entirely similar to the calculation of the canonical state $\psi$ in \cite[Prop.~5.5]{MR4945955} and the reflection positivity axiom \ref{LTO:RP} for Levin-Wen in \S\ref{sec:LevinWen} below.
\end{ex}

\section{A Haag duality axiom for LTOs}
\label{sec:hdlto}

We continue our assumption here that $(\fA,p)$ is an LTO with projections assigned to disk-like regions $R$, which are extended to cone-like regions by looking at the net of von Neumann algebras $\Lambda\mapsto \cA(\Lambda)\coloneqq\fA(\Lambda)''$ acting on $L^2(\fA,\psi)$.
As before, we write $\cB(\Lambda)\coloneqq p_\Lambda \cA(\Lambda) p_\Lambda$, which is the \emph{boundary algebra} along $\partial \Lambda$ (cf.~\cite[\S5.4]{MR4945955}).

The {\it Haag duality} axiom for LTOs \ref{LTO:HD} is a compatibility condition for the boundary algebra of a cone-shaped region $\Lambda$ and its complement $\Lambda'$.
As our axiom treats $\Lambda$ and $\Lambda'$ with equal footing, we write $\Lambda_+=\Lambda$ and $\Lambda_-=\Lambda'$ in this section.
Our axiom below is inspired by~\cite[Thm.~4.1.3]{2509.23734}.

\begin{axiom}
The Haag duality axiom for LTOs states:
\begin{enumerate}[label=\textup{(LTO-HD)}]
\item 
\label{LTO:HD}
For any cone $\Lambda_+$ with complement $\Lambda_-$, whenever $R \ll S$ are sufficiently large disk-like regions such that both
$R_\pm\coloneqq R\cap \Lambda_\pm$ and $S_\pm\coloneqq S\cap \Lambda_\pm$ are also sufficiently large disk-like regions with $R_\pm\Subset S_\pm$,
and $\partial[\Lambda_\pm]$ intersects both $\partial[R]$ and $\partial[S]$ transversely,
then
$
p_{S_+} \fA(R_+) p_S
=
p_{S_+}p_{S_-} \fA(R) p_S
=
p_{S_-} \fA(R_-) p_S
$,
i.e.,
\[
\fB(R_+\Subset S_+)p_S
=
p_{S_+}p_{S_-} \fA(R) p_S
=
\fB(R_-\Subset S_-)p_S.
\]
\end{enumerate}
\end{axiom}

The left hand side below is a cartoon of this general setup.
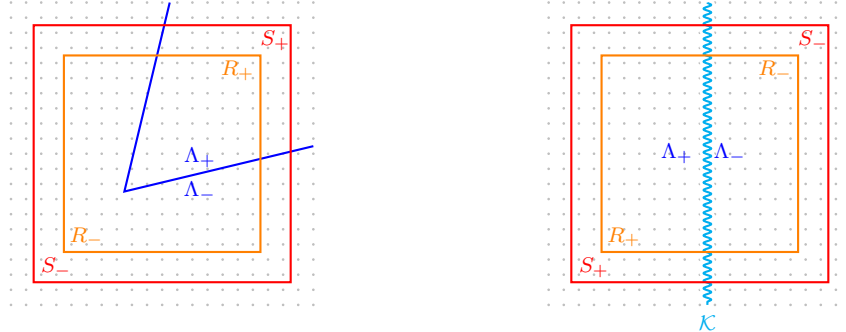
\[
\tikzmath{
\foreach \x in {-10,...,10}{
\foreach \y in {-10,...,10}{
\filldraw[gray!50!white] ($ .2*(\x,\y) $) circle (.01cm);
}}
\draw[thick, blue] (0.1,2) -- (-.5,-.5) -- (2,.1);
\node[blue] at (.5,-.05) {$\scriptstyle \Lambda_+$};
\node[blue] at (.5,-.5) {$\scriptstyle \Lambda_-$};
\draw[thick, red] (-1.7,-1.7) rectangle (1.7,1.7);
\node[red] at (1.5,1.5) {$\scriptstyle S_+$};
\node[red] at (-1.4,-1.5) {$\scriptstyle S_-$};
\draw[thick, orange] (-1.3,-1.3) rectangle (1.3,1.3);
\node[orange] at (1,1.1) {$\scriptstyle R_+$};
\node[orange] at (-1,-1.1) {$\scriptstyle R_-$};
}
\qquad\qquad\qquad\qquad
\tikzmath{
\foreach \x in {-10,...,10}{
\foreach \y in {-10,...,10}{
\filldraw[gray!50!white] ($ .2*(\x,\y) $) circle (.01cm);
}}
\draw[thick, snake, knot, cyan] (0.1,2) node[above]{$\scriptstyle \phantom{\cK}$} -- (.1,-2) node[below]{$\scriptstyle \cK$};
\node[blue] at (.4,0) {$\scriptstyle \Lambda_-$};
\node[blue] at (-.3,0) {$\scriptstyle \Lambda_+$};
\draw[thick, red] (-1.7,-1.7) rectangle (1.7,1.7);
\node[red] at (1.5,1.5) {$\scriptstyle S_-$};
\node[red] at (-1.4,-1.5) {$\scriptstyle S_+$};
\draw[thick, orange] (-1.3,-1.3) rectangle (1.3,1.3);
\node[orange] at (1,1.1) {$\scriptstyle R_-$};
\node[orange] at (-1,-1.1) {$\scriptstyle R_+$};
}
\]
The right hand side above illustrates the special case when our complementary cones $\Lambda_\pm$ are complementary half-spaces $\bbH_\pm$ determined a hyperplane $\cK$ that is normal to one of the standard coordinate axes in our $\bbZ^n$ lattice.
We specialize to this case for the remainder of the section for ease of exposition.
Our main results are equally valid for more general cone-like regions.

\begin{rem}
In the case that our LTO $(\fA,p)$ comes from a quantum spin system, $\fA(R)=\fA(R_+)\otimes \fA(R_-)$, and the condition
\[
\fB(R_+\Subset S_+)p_S
=
p_{S_+}p_{S_-} \fA(R) p_S
=
\fB(R_-\Subset S_-)p_S
\]
reduces to the condition
\[
\fB(R_+\Subset S_+)p_S
=
\fB(R_-\Subset S_-)p_S.
\]
\end{rem}

Intuitively, \ref{LTO:HD} says that when restricted to the ground state space for $\Delta$, we can either act with the boundary algebra `left' of the cut, or equivalently with the boundary algebra from the `right'.
The ability to create excitations at the boundary either from the left or the right plays a crucial role in the proof of Haag duality in~\cite{MR2956822,MR3426207,2509.23734}.

\subsection{Examples}

\begin{ex}
\label{ex:ToricCodeHD}
A complete description of the boundary algebras for the Toric Code was given in \cite[\S3]{MR4945955}.
There are two scenarios depending on whether the boundary is rough or smooth, and a $\bbZ/2$ flip symmetry will naturally swap a rough boundary with a smooth boundary.
The boundary algebras are supported on sites within distance one to the cut; we denote the sites on the left $(+)$ side by $I$ and the sites on the right $(-)$ side by $J$.
The generators for the left and right sides are given as follows.
\begin{equation}
\label{eq:TC-BoundaryExcitations}
\tikzmath{
\draw (1.5,.05) -- (1.5,3.7);
\foreach \y in {.75,1.5,2.25,3}{
    \draw (2.25,\y) -- (1.5,\y);
}
\draw[thick, red] (2.25,1.5) -- (1.5,1.5) -- (1.5,2.25) -- (2.25,2.25);
\node[red] at (1.3,1.875) {$\scriptstyle X$};
\node[red] at (1.875,1.3) {$\scriptstyle X$};
\node[red] at (1.875,2.45) {$\scriptstyle X$};
\node[red] at (2.5,1.875) {$D_p$};
\node[red] at (1.875,1.875) {$\scriptstyle p$};
\draw[thick, orange] (1.5,3) -- (2.25,3);
\node[orange] at (1.875,3.2) {$\scriptstyle Z$};
\node[orange] at (2.5,3) {$C_\ell$};
\node at (2,-.3) {$\scriptstyle\fB(R_+\Subset S_+)\coloneqq 
{\rm C^*}\set{C_\ell,D_p}{\ell,p\subset I}$};
\draw[thick, DarkGreen, rounded corners=5pt] (1,3.7) -- (2.25,3.7) -- (2.25,.05) -- (1,.05);
\node[DarkGreen] at (1.25,3.5) {$R_+$}; 
}
\qquad
\tikzmath{
\draw[scale=.75] (-2.5,-2.5) grid (2.5,2.5);
\foreach \x in {-1.5,-.75,0,.75, 1.5}{
\foreach \y in {-1.125,-.375,.375,1.125}{
\filldraw (\x,\y) circle (.05cm);
}}
\foreach \x in {-1.125,-.375,.375,1.125}{
\foreach \y in {-1.5,-.75,0,.75, 1.5}{
\filldraw (\x,\y) circle (.05cm);
}}
\node[DarkGreen] at (-.5,-2.1) {$\scriptstyle I$};
\foreach \y in {-1.5,-.75,0,.75, 1.5}{
\filldraw[DarkGreen] (-.375,\y) circle (.05cm);
}
\foreach \y in {-1.125,-.375,.375,1.125}{
\filldraw[DarkGreen] (-.75,\y) circle (.05cm);
}
\node[blue] at (.25,-2.1) {$\scriptstyle J$};
\foreach \y in {-1.5,-.75,0,.75, 1.5}{
\filldraw[blue] (.375,\y) circle (.05cm);
}
\foreach \y in {-1.125,-.375,.375,1.125}{
\filldraw[blue] (0,\y) circle (.05cm);
}
\draw[thick, cyan, snake] (-.1725,-1.875) node[below]{$\scriptstyle \cK$}-- (-.1725,1.875);
}
\qquad
\tikzmath{
\draw (-2.25,.05) -- (-2.25,3.7);
\foreach \y in {.75,1.5,2.25,3}{
    \draw (-2.25,\y) -- (-1.5,\y);
}
\draw[thick, red] (-2.25,.75) -- (-2.25,1.5);
\node[red] at (-2.05,1.175) {$\scriptstyle X$};
\node[red] at (-2.6,1.175) {$D_\ell$};
\draw[thick, orange] (-1.5,3) -- (-2.25,3);
\draw[thick, orange] (-2.25,2.25) -- (-2.25,3.7);
\node[orange] at (-1.5,3.1) {$\scriptstyle Z$};
\node[orange] at (-2,2.675) {$\scriptstyle Z$};
\node[orange] at (-2,3.375) {$\scriptstyle Z$};
\node[orange] at (-2.6,3) {$C_s$};
\node at (-2,-.3) {$\scriptstyle\fB(R_-\Subset S_-)\coloneqq 
{\rm C^*}\set{C_s,D_\ell}{\ell,s\subset J}$};
\draw[thick, blue, rounded corners=5pt] (-1.25,3.7) -- (-2.35,3.7) -- (-2.35,.05) -- (-1.25,.05);
\node[blue] at (-1.5,3.5) {$R_-$}; 
}
\end{equation}
It follows immediately from the descriptions of the boundary algebras in \eqref{eq:TC-BoundaryExcitations} that the Toric Code satisfies \ref{LTO:HD}.
Indeed, for each generator $x_+\in \fB(R_+\Subset S_+)$, there is a term $t$ of the Hamiltonian that is absorbed by $p_S$ such that $x_+t=x_-t$ for some generator $x_-\in \fB(R_-\Subset S_-)$.
Said another way, for stars $s$ and plaquettes $p$ contained in $R$ that intersect the cut $\cK$, we see that $A_s=C_\ell\otimes C_s$ and $B_p=D_p\otimes D_\ell$ as below.
\[
A_s
=
\tikzmath{
\draw (-1,0) -- (1,0);
\draw (0,-1) -- (0,1);
\draw[thick, cyan, snake] (-.25,-1) -- (-.25,1);
\filldraw[DarkGreen] (-.5,0) circle (.05cm);
\node[DarkGreen] at (-.5,.2) {$\scriptstyle Z$};
\node[DarkGreen] at (-.5,-.2) {$\scriptstyle \ell$};
\filldraw[blue] (.5,0) circle (.05cm);
\node[blue] at (.5,.2) {$\scriptstyle Z$};
\filldraw[blue] (0,.5) circle (.05cm);
\node[blue] at (.2,.5) {$\scriptstyle Z$};
\filldraw[blue] (0,-.5) circle (.05cm);
\node[blue] at (.2,-.5) {$\scriptstyle Z$};
}
\qquad\qquad\qquad
B_p
=
\tikzmath{
\draw (0,0) rectangle (1,1);
\draw[thick, cyan, snake] (.75,-.5) -- (.75,1.5);
\filldraw[DarkGreen] (0,.5) circle (.05cm);
\node[DarkGreen] at (-.2,.5) {$\scriptstyle X$};
\filldraw[DarkGreen] (.5,0) circle (.05cm);
\node[DarkGreen] at (.5,-.2) {$\scriptstyle X$};
\filldraw[DarkGreen] (.5,1) circle (.05cm);
\node[DarkGreen] at (.5,1.2) {$\scriptstyle X$};
\filldraw[blue] (1,.5) circle (.05cm);
\node[blue] at (.85,.5) {$\scriptstyle \ell$};
\node[blue] at (1.2,.5) {$\scriptstyle X$};
}
\]
\end{ex}

\begin{ex}
\label{ex:QuantumDoubleHD}
A complete description of the boundary algebras for Kitaev's Quantum Double Model for finite groups $G$ was given in \cite[\S3]{MR4814524}.
Again, there are two scenarios depending on whether the boundary is rough or smooth, and again a $\bbZ/2$ flip symmetry will swap the two.
The boundary algebras are supported on sites within distance one to the cut; we denote the sites on the left $(+)$ side by $I$ and the sites on the right $(-)$ side by $J$.
The generators for the left and right sides are given as follows.\footnote{In \cite{MR4814524}, generators for the boundary algebra were shown to be partial plaquette operators.
The formulae for the generators on the right-hand side for $\fB(R_-\Subset S_-)$ appear different from those for the smooth boundary algebra in \cite{MR4814524} as the partial plaquette and star operators for a right-facing boundary are different from the partial plaquette and star operators for the left-facing boundary.
} 
\[
\tikzmath{
\draw (1.5,.05) -- (1.5,3.7);
\foreach \y in {.75,1.5,2.25,3}{
    \draw (2.25,\y) -- (1.5,\y);
}
\draw[thick, red] (2.25,1.5) -- (1.5,1.5) -- (1.5,2.25) -- (2.25,2.25);
\foreach \x in {1.875}{
\foreach \y in {.75, 1.5, 2.25, 3}{
\filldraw (\x, \y) circle (0.05cm);
}}
\foreach \x in {1.5}{
\foreach \y in {.375, 1.125, 1.875, 2.625, 3.375}{
\filldraw (\x, \y) circle (0.05cm);
}}
\foreach \x in {1.875}{
\foreach \y in {1.5, 2.25}{
\filldraw[red] (\x,\y) circle (.05cm);
}}
\foreach \x in {1.5}{
\foreach \y in {1.875}{
\filldraw[red] (\x,\y) circle (.05cm);
}}
\foreach \x in {1.875}{
\foreach \y in {3}{
\filldraw[orange] (\x,\y) circle (.05cm);
}}
\node[red][scale = 0.8] at (1.15,1.875) {$R_{g^{-1}}$};
\node[red][scale = 0.8] at (1.875,1.3) {$L_g$};
\node[red][scale = 0.8] at (1.875,2.5) {$R_{g^{-1}}$};
\node[red] at (2.7,1.875) {$Q_p^{(g)}$};
\node[red] at (1.875,1.875) {$\scriptstyle p$};
\draw[thick, orange] (1.5,3) -- (2.25,3);
\node[orange] at (1.875,3.2) {$\scriptstyle P_g$};
\node[orange] at (2.7,3) {$P_\ell^{g}$};
\node at (2,-.5) {$\scriptstyle\fB(R_+\Subset S_+)\coloneqq 
\rmC^*\set{P_\ell^{g},Q_p^{(g)}}{ \ell,p\subset I,\, g\in G}$};
\draw[thick, DarkGreen, rounded corners=5pt] (1,3.7) -- (2.25,3.7) -- (2.25,.05) -- (1,.05);
\node[DarkGreen] at (1.25,3.5) {$R_+$}; 
}
\qquad\qquad
\tikzmath{
\draw (-2.25,.05) -- (-2.25,3.7);
\foreach \y in {.75,1.5,2.25,3}{
    \draw (-2.25,\y) -- (-1.5,\y);
}
\foreach \x in {-1.875}{
\foreach \y in {.75, 1.5, 2.25, 3}{
\filldraw (\x, \y) circle (0.05cm);
}}
\foreach \x in {-2.25}{
\foreach \y in {.375, 1.125, 1.875, 2.625, 3.375}{
\filldraw (\x, \y) circle (0.05cm);
}}
\foreach  \x in {-2.25}{
\foreach \y in {1.125}{
\filldraw[red] (\x,\y) circle (.05cm);
}}
\foreach  \x in {-2.25}{
\foreach \y in {2.625, 3.375}{
\filldraw[orange] (\x,\y) circle (.05cm);
}}
\foreach  \x in {-1.875}{
\foreach \y in {3}{
\filldraw[orange] (\x,\y) circle (.05cm);
}}
\draw[thick, red] (-2.25,.75) -- (-2.25,1.5);
\node[orange] at (-3,3) {$\scriptstyle \sum\limits_{hk\ell^{-1} = g}$};
\node[orange] at (-2.05,2.575) {$\scriptstyle P_{\ell}$};
\node[orange] at (-2.05,3.375) {$\scriptstyle P_{h}$};
\node[orange] at (-1.85,2.85) {$\scriptstyle P_{k}$};
\node[red] at (-2,1.175) {$\scriptstyle L_g$};
\node[red] at (-2.8,1.175) {$L_\ell^{g}$};
\draw[thick, orange] (-1.5,3) -- (-2.25,3);
\draw[thick, orange] (-2.25,2.25) -- (-2.25,3.7);
\node[orange] at (-4.2,3) {$S_s^{(g)}=$};
\node at (-2,-.4) {$\scriptstyle\fB(R_-\Subset S_-)\coloneqq 
\rmC^*\set{S_s^{(g)},L_\ell^{g}}{ \ell,s\subset J,\,g\in G}$};
\draw[thick, blue, rounded corners=5pt] (-1.25,3.7) -- (-2.35,3.7) -- (-2.35,.05) -- (-1.25,.05);
\node[blue] at (-1.5,3.5) {$R_-$}; 
}
\]
Above,
$R_g$, $L_g$, and $P_g$ act on $\bbC^G = \operatorname{span}\set{|g \rangle}{g \in G}$ by $R_g |h \rangle \coloneqq |hg \rangle$, $L_g |h \rangle \coloneqq |gh \rangle$, and $P_g |h \rangle \coloneqq \delta_{g=h} |h \rangle$.

Again, for each generator $x_+\in \fB(R_+\Subset S_+)$, there is a term $t$ of the Hamiltonian that is absorbed by $p_S$ such that $x_+t=x_-t$ for some generator $x_-\in \fB(R_-\Subset S_-)$.
As before, for stars $s$ and plaquettes $p$ contained in $R$ that intersect the cut $\cK$, we see that $P_\ell^{g} A_s = S^{(g)}_s A_s$ and $Q_p^{(g^{-1})}B_p^g = L_\ell^g $.
\end{ex}

\begin{ex}
We will give a simple argument in~\S\ref{sec:LevinWen} below that the Levin-Wen and Walker-Wang models satisfy the stronger reflection positivity axiom for LTOs \ref{LTO:RP} for $\bbZ/2$-symmetric regions, which implies \ref{LTO:HD} for $\Lambda_\pm=\bbH_\pm$.
The general case for \ref{LTO:HD} for arbitrary cone-like regions follows similarly.
\end{ex}

\subsection{Haag duality for the net of cone von Neumann algebras}

In this section, we consider the standard half-spaces $\bbH_\pm$ on the $\pm$ side of a hyperplane $\cK$ which bisects our lattice $\cL$.
We write $\cB_\pm=p_{\bbH_{\pm}} \cA(\bbH_{\pm}) p_{\bbH_{\pm}}$.
Assuming \ref{LTO:CentralSupportOne} and \ref{LTO:FaithfulOnBoundary}, 
$\cB_\pm$ is isomorphic to the von Neumann algebra generated by $\fB_\pm$ on $L^2(\fB_\pm,\psi_{\fB_\pm})$.
We are now in a position that we can prove Theorem~\ref{thm:Main} from the introduction.

\begin{thm}
\label{thm:HD}
Suppose the LTO $(\fA,p)$ satisfies \ref{LTO:HD} along with:
\begin{itemize}
\item 
\ref{LTO:CentralSupportOne} for $\bbH_\pm$: each $p_{\bbH_\pm}$ has central support 1 in $\cA(\bbH_\pm)$, and 
\item 
\ref{LTO:FaithfulOnBoundary} for $\bbH_\pm$: the canonical state $\psi$ is faithful on $\cB_\pm$.
\end{itemize}
Assume moreover that $\fA(R)$ is finite dimensional whenever $R$ is finite.
Then $\cA(\bbH_\pm)' = \cA(\bbH_\mp)$.
\end{thm}

We remind the reader that this result is equally valid for arbitrary cone-like regions; we only specialize to $\Lambda=\bbH_+$ for convenience.
We begin with the following observation distilled from the proof of \cite[Lem.~2.1.4]{2509.23734}.

\begin{fact}
\label{fact:CompressionCommutant}
When $M,N\subset B(H)$ are von Neumann algebras with $N\subseteq M'$ and 
$p\in M$ has central support one, then the following are equivalent:
\begin{enumerate}[label=\textup{($p$\arabic*)}]
\item
\label{p:Commutant} 
$N=M'$, or equivalently, $N'=M$,
\item
\label{p:Compression}
$(Np)'=pN'p = pMp$, or equivalently, $Np=M'p=(pMp)'$, acting on $pH$.
\end{enumerate}
Indeed, \ref{p:Commutant} clearly implies \ref{p:Compression}, and since the map $M'\to M'p$ is injective, \ref{p:Compression} implies \ref{p:Commutant}.

Now if $q\in N$ also has central support one, then 
$pq$ has central support $q$ in $Mq$.
As an immediate corollary, 
$N=M'$ if and only if $pMp=(Np)'$
acting on $pH$, if and only if $pMpq = (qNpq)'$ acting on $pqH$.
\end{fact}

\begin{proof}[Proof of Theorem \ref{thm:HD}]
Since each $p_{\bbH_{\pm}}$ has central support $1$ in $\cA(\bbH_{\pm})$, by Fact \ref{fact:CompressionCommutant}, the statement is equivalent to
\begin{equation}
\label{eq:CommutantToShow}
(p_{\bbH_+} \cA(\bbH_+) p_{\bbH_+} p_{\bbH_-})' 
\overset{?}{=} 
p_{\bbH_-} \cA(\bbH_-) p_{\bbH_-} p_{\bbH_+}
\end{equation}
acting on the compressed Hilbert space $\cH\coloneqq p_{\bbH_+} p_{\bbH_-} L^2(\fA, \psi)$.
Since $p_\Delta\Omega = \Omega$ for every region $\Delta$,
$\Omega$ is cyclic for $p_{\bbH_{\pm}} \cA(\bbH_{\pm}) p_{\bbH_{\pm}} p_{\bbH_{\mp}}=\cB_\pm p_{\bbH_{\mp}}$
acting on $\cH$ by \ref{LTO:HD}. 
Again, since $p_{\bbH_{\pm}}\Omega=\Omega$, $\Omega$ is also separating by faithfulness of $\psi$.
Since $p_{\bbH_{\pm}}$ has central support 1, we see that 
$p_{\bbH_{\pm}} \cA(\bbH_{\pm}) p_{\bbH_{\pm}} p_{\bbH_{\mp}} \cong \cB_\pm$ acting on $\cH \cong L^2(\cB_\pm, \psi)$.
Under this identification,  proving \eqref{eq:CommutantToShow} above is equivalent to proving
$\cB_+=\cB_-'$ on $\cH$.

Consider a boundary interval $I$.
For every $\bbZ/2$-flip symmetric region $R$ with $I=R \cap \cK$, we can use \ref{LTO:HD} to identify both of the finite dimensional subspaces $L^2(\fB_\pm(I),\psi)\subset L^2(\fB_\pm, \psi)$ 
with the subspace
\[
p_{\bbH_+} p_{\bbH_-} \overline{\fA(R)\Omega}\subset 
p_{\bbH_+} p_{\bbH_-} L^2(\fA, \psi)=\cH.
\]
We will simply write $L^2I \subset \cH$ for this common subspace, which contains $\Omega$.
Observe that $L^2I$ carries a left $\fB_+(I)$-action and a commuting right $\fB_-(I)$-action, and $\Omega$ is cyclic and separating for both these actions.
By \ref{LTO:HD}, these algebras are each other's commutants.

Let $\Delta, J$ be canonical Tomita-Takesaki operators for the von Neumann algebra $\cB_+$
acting on $L^2(\cB_+,\psi)\cong \cH$.
We claim that $\Delta,J$ restrict to $\Delta_I, J_I$, the canonical Tomita-Takesaki operators for the von Neumann algebra $\fB_+(I)$ on $L^2(\fB_+(I),\psi)\cong L^2I$.
By Corollary \ref{cor:ModularGroupPreservesSubalgebras}, $\sigma^\psi$ preserves $\fB_{+}(I)$ for each region $I\subset \cK$, so $\Delta^{it}$ commutes with each projection $p_I$ onto $L^2(\fB_{+}(I), \psi)$. 
Since $\Delta$ is affiliated with the von Neumann algebra generated by $\set{\Delta^{it}}{t \in \bbR}$ by Stone's Theorem, $\Delta$ commutes with $p_I$, so $\Delta$ preserves each $L^2(\fB_{+}(I), \psi)$. 
Therefore, $J = S_\psi \Delta^{-1/2}$ also preserves $L^2(\fB_{+}(I), \psi)$,
as does $F_\psi = J\Delta^{-1/2}$.
It follows that $\Delta$ and $J$ restrict to $\Delta_I$ and $J_I$, defined on $L^2(\fB_+(I),\psi|_{\fB_+(I)})$.

Now under the identification $L^2(\fB_+,\psi)\cong \cH$ and $L^2(\fB_+(I),\psi)\cong L^2I$,
we have 
\[
J\fB_+(I)J = J_I\fB_+(I)J_I = \fB_-(I),
\]
which immediately implies
\[
J\cB_+ J 
= 
J\left(\varinjlim \fB_+(I)\right)''J
=
\left(\varinjlim J_I\fB_+(I)J_I\right)''
=
\left(\varinjlim \fB_-(I)\right)''
=
\cB_-''
\]
as desired.
\end{proof}

\begin{rem}
Observe that the proof of Theorem \ref{thm:HD} did not really require the local algebras $\fA(R)$ to be finite dimensional, but rather the local boundary algebras $\fB(I)$ needed to be finite dimensional. 
\end{rem}

\section{A reflection positivity axiom for LTOs}

Throughout this section we again assume that for an LTO $(\fA, p)$, the local algebra $\fA(R)$ is finite dimensional if $R$ is finite
(although we can get by if only the boundary algebras are finite dimensional as above).

Suppose we have a codimension 1 hyperplane $\cK$ that divides our lattice $\cL$ into two halves with a $\bbZ/2$ reflection symmetry $\theta$.
Set $\bbH_\pm$ to be the regions on either side of the hyperplane, giving local operator algebras $\fA(\bbH_\pm)$.
We assume we have an involution (conjugate-linear $*$-isomorphism of period 2) $\Theta:\fA\to \fA$ such that $\Theta(\fA(R))=\fA(\theta R)$ and $\Theta(p_{R}) = p_{\theta R}$ for every $R\subset \cL$.
In particular, $\Theta$ gives a conjugate-linear algebra isomorphism $\fA(\bbH_+)\to \fA(\bbH_-)$.

\begin{lem}
The map $u_\Theta:x\Omega \mapsto \Theta(x)\Omega$ is antiunitary on $L^2(\fA,\psi)$.
\end{lem}
\begin{proof}
First, suppose $R\ll S$ and $x\in \fA(R)$ so that
$\psi(x)p_S=p_Sxp_S$.
Since $\theta R\ll \theta S$, we have
\[
\psi(\Theta(x))p_{\theta S}
=
p_{\theta S}\Theta(x)p_{\theta S}
=
\Theta(p_S)\Theta(x)\Theta(p_S)
=
\Theta(p_Sxp_S)
=
\Theta(\psi(x)p_S)
=
\overline{\psi(x)}p_{\theta S}.
\]
Hence $\psi(\Theta(y))=\overline{\psi(y)}$ for all $y\in \fA$.
We now see
\[
\|\Theta(x)\Omega\|_2^2
=
\psi(\Theta(x^\dag x))
=
\psi(x^\dag x)
=
\|x\Omega\|_2^2.
\]
Thus $u_\Theta$ is an invertible anti-linear isometry, i.e., an antiunitary.
\end{proof}

For a $\theta$-symmetric region $R=\theta R$, we write $R_\pm \coloneqq  R \cap \bbH_\pm$.
Observe that if $R,S$ are $\theta$-symmetric regions with $R \ll S$, then $R_\pm \Subset S_\pm$.
The {\it reflection positivity} axiom for LTOs is a {\it local} compatibility condition with $\theta$.

\begin{axiom}
The reflection positivity axiom for LTOs states:
\begin{enumerate}[label=\textup{(LTO-RP)}]
\item 
\label{LTO:RP}
The canonical state $\psi$ is faithful on $\fB$, and
whenever $R \ll S$ are sufficiently large disk-like $\theta$-symmetric regions such that $\partial[R]$ and $\partial[S]$ are transverse to $\cK$,
\[
\Theta(\sigma^\psi_{-i/2}(x^\dag)) p_S
=
x p_S
\qquad\qquad\qquad\qquad
\forall\,x\in \fB(R_+\Subset S_+)
\]
where $\sigma^\psi$ is the modular automorphism group of the canonical state.
Note that the above is well-defined as the boundary algebra $\fB(R_+\Subset S_+)$ is finite dimensional, so all elements are analytic.
\end{enumerate}
\end{axiom}

\begin{rem}
Clearly \ref{LTO:RP} implies \ref{LTO:HD} for LTOs coming from quantum spin systems for $\Lambda_\pm=\bbH_\pm$.
For abstract quantum spin systems, we do not necessarily know that
$$
p_{S_+}p_{S_-} \fA(R)p_S
\overset{?}{=}
p_{S_+}p_{S_-} \fA(R_+)\fA(R_-)p_S.
$$
\end{rem}

\begin{rem}
The reader might wonder what role the (non $*$-preserving!) algebra map $\sigma^\psi_{-i/2}$ plays in the so-called `reflection positivity' axiom \ref{LTO:RP}.
We have a nice description in the setting where 
each $p_{\bbH_\pm}$ has central support one in $\cA(\bbH_\pm)$, $\psi$ is faithful on the boundary (von Neumann) algebra $\cB(\bbH_\pm)=p_{\bbH_\pm} \cA(\bbH_\pm)p_{\bbH_\pm}$, and \ref{LTO:HD} holds for $\Lambda_\pm=\bbH_\pm$.
Recall from the proof of Theorem~\ref{thm:HD} that 
we have a compatibility of the modular group $\sigma^\psi$ for $\cB_+=\fB''_+$ on $L^2(\fA,\psi)$ and for each $\fB(I)$.
Writing $\Omega \in L^2(\cB_+,\psi)$ for the image of $1$, observe that for $I=R\cap \cK$ for $\bbZ/2$-symmetric regions $R\ll S$,
$$
u_\Theta x\Omega
=
u_\Theta x p_S \Omega
=
\sigma^\psi_{-i/2}(x^\dag) p_S \Omega
=
\sigma^\psi_{-i/2}(x^\dag) \Omega
=
J_\psi x \Omega
\qquad\qquad
\forall\,
x\in \fB_+(I)=\fB(R_+\Subset S_+).
$$
Hence the reflection positivity axiom for LTOs is telling us that $u_\Theta=J_\psi$ on 
$$
p_{\bbH_+}p_{\bbH_-}L^2(\fA,\psi)
\cong L^2(\fB_+,\psi).
$$
\end{rem}

\begin{ex}
Consider the Toric Code as in~\S\ref{sec:hdlto}, where for convenience we now rotate the lattice by $45^{\circ}$, and choose a horizontal hyperplane $\cK$ (alternatively, we can take a diagonal hyperplane $\cK$ in the original model):
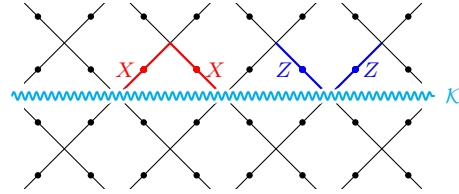
\[
\tikzmath[scale=.7]{
\begin{scope}
    \clip (.25,2.25) rectangle (7.75, 5.75);
\foreach \x in {-8,-6,-4,-2,0,2,4,6,8}{
    \draw (\x,0) -- ++(8,8);
    \draw (\x,8) -- ++(8,-8);
}

\node at (2.15,4.5) {\color{red}$\scriptstyle X$};
\node at (3.85,4.5) {\color{red}$\scriptstyle X$};
\node at (5.15,4.5) {\color{blue}$\scriptstyle Z$};
\node at (6.8,4.5) {\color{blue}$\scriptstyle Z$};
\draw[thick,red] (2,4) -- (3,5) -- (4,4);
\draw[thick,blue] (5,5) -- (6,4) -- (7,5);
\foreach \x in {0.5,...,8.5}{
    \foreach \y in {0.5,...,6.5}{
        \filldraw[black] (\x,\y) circle (.05cm);
    }
\filldraw[red] (2.5,4.5) circle (.05cm);
\filldraw[red] (3.5,4.5) circle (.05cm);
\filldraw[blue] (5.5,4.5) circle (.05cm);
\filldraw[blue] (6.5,4.5) circle (.05cm);
}
\end{scope}
\draw[thick, cyan, snake, knot] (0,4) -- (8,4) node[right, cyan]{$\scriptstyle \mathcal{K}$};
}
\]
Let $\theta$ be the reflection through the hyperplane $\mathcal{K}$ and write $\widehat{\theta}$ for the automorphism of $\fA$ it induces.
Define the anti-unitary operator $K: \mathbb{C}^2 \to \mathbb{C}^2$ as complex conjugation with respect to the standard basis.
Write $k$ for the conjugate-linear automorphism of $\fA$ obtained by conjugating with $K$ at each site.
Then $\Theta = \widehat{\theta} \circ k$ is a conjugate-linear $*$-automorphism of period 2 such that $\Theta(\fA(R)) = \fA(\theta R)$.
We claim that the Toric Code satisfies~\ref{LTO:RP} for $\Theta$ (cf.~\cite[\S 6.1]{2510.20662}).

Choose $\theta$-symmetric regions $R \ll S$.
We write $R_+$ (resp. $S_+$) for the part above the cut $\cK$.
Observe that $\Theta(x_j^\dagger) = x_{\theta j}$ for $x \in \fA(\{j\})$. 
It follows that $\Theta(p_S) = p_{\theta S}$.
By adapting~\cite[Algorithm~3.10]{MR4945955} it follows that $\fB(R_+ \Subset S_+)$ is generated by the operators $Z \otimes Z$ and $X \otimes X$ localized on the \emph{partial} stars and plaquettes near the boundary, as indicated by the blue and red edges respectively in the figure above.
Observe that if $j, j+1$ are the two edges of a star $s$ straddling the cut $\mathcal{K}$, we have that $\Theta(Z_j Z_{j+1}) A_s = Z_j Z_{j+1}$, with a similar statement for the partial plaquette operators.
Because the canonical boundary state $\psi$ is a trace for the Toric Code, and hence the modular automorphisms act trivially, the claim readily follows.
\end{ex}

In the remainder of this section, we prove \ref{LTO:RP} for the Levin-Wen model based on a unitary fusion category (UFC) $\cX$.
We equip $\cX$ with its unique unitary spherical structure and spherical tracial state $\tr_\cX$ \cite{MR2091457,MR4133163}, and we make heavy use of the graphical calculus \cite[\S2.5]{MR3663592}.
We write $\Irr(\cX)$ for a set of representatives of the simple objects of $\cX$, and for $x\in\Irr(\cX)$, we write $d_x\coloneqq \tr_\cX(\id_x)$ for its quantum dimension.

\subsection{Skein modules}

Recall that the skein module $\cS_\cX(\bbD, n)$ for the UFC $\cX$ on a disk with $n$ boundary points is the Hilbert space
\[
    \cS_\cX(\bbD, n) \coloneqq  \cX\left( 1_\cX \to X^{\otimes n}\right) = \bigoplus_{x_1, \dots, x_n \in \Irr(\cX)} \cX(1_{\cX} \to x_1 \otimes \cdots \otimes x_n),
\]
where $X \coloneqq  \bigoplus_{x \in \Irr(\cX)} x$ and all direct sums are orthogonal.
The inner product is given by
\[
\left\langle
\tikzmath{
\draw (-.3,.3) --node[left]{$\scriptstyle x_1$} (-.3,.7);
\node at (.05,.5) {$\cdots$};
\draw (.3,.3) --node[right]{$\scriptstyle x_n$} (.3,.7);
\roundNbox{}{(0,0)}{.3}{.3}{.3}{$f$}
}
\middle|
\tikzmath{
\draw (-.3,.3) --node[left]{$\scriptstyle y_1$} (-.3,.7);
\node at (.05,.5) {$\cdots$};
\draw (.3,.3) --node[right]{$\scriptstyle y_n$} (.3,.7);
\roundNbox{}{(0,0)}{.3}{.3}{.3}{$g$}
}
\right\rangle
\coloneqq 
\prod_{j=1}^n \delta_{x_j=y_j} \frac{1}{\sqrt{d_{x_j}}}
\cdot
\tikzmath{
\draw (-.3,.3) --node[left]{$\scriptstyle x_1$} (-.3,.7);
\node at (.05,.5) {$\cdots$};
\draw (.3,.3) --node[right]{$\scriptstyle x_n$} (.3,.7);
\roundNbox{}{(0,0)}{.3}{.3}{.3}{$g$}
\roundNbox{}{(0,1)}{.3}{.3}{.3}{$f^\dag$}
}\,.
\]
We will make great use of the {\it fusion/semisimplicity relation} for UFCs;
that is, 
\begin{equation}
\label{eq:Fusion}
\sum_{z\in\Irr(\cX)} \sqrt{d_z}\cdot
\tikzmath{
\draw (-.3,-.6) node[below]{$\scriptstyle x$} -- (0,-.3) -- (.3,-.6) node[below]{$\scriptstyle y$};
\draw (-.3,.6) node[above]{$\scriptstyle x$} -- (0,.3) -- (.3,.6) node[above]{$\scriptstyle y$};
\draw (0,-.3) --node[left]{$\scriptstyle z$} (0,.3);
\filldraw[red] (0,-.3) circle (.05cm);
\filldraw[red] (0,.3) circle (.05cm);
}
=
\sqrt{d_xd_y}\cdot
\tikzmath{
\draw (-.3,-.6) node[below]{$\scriptstyle x$} --(-.3,.6) node[above]{$\scriptstyle x$};
\draw (.3,-.6) node[below]{$\scriptstyle y$} --(.3,.6) node[above]{$\scriptstyle y$};
}
\end{equation}
where as in \cite[\S2.5]{MR3663592}, the pair of shaded nodes above means summing over an orthonormal basis (ONB) of the trivalent skein module $\cS_\cX(\bbD,3)$.

We now analyze the skein module $\cH_n\coloneqq \cS_\cX(\bbD,2n)$ as a standard form for $\fB_n\coloneqq \End_\cX(X^{\otimes n})$.
For $\varphi\in \cX(a_1\otimes \cdots \otimes a_n\to b_1\otimes \cdots \otimes b_n)$
where all $a_j,b_j\in \Irr(\cX)$,
we have the gluing operator $\Gamma_\varphi$ on $\cH_n$ given by
\[
\tikzmath{
\draw (-.3,.3) --node[left]{$\scriptstyle y_1$} (-.3,.7);
\draw (-.3,-.3) --node[left]{$\scriptstyle x_1$} (-.3,-.7);
\node at (.05,.5) {$\cdots$};
\node at (.05,-.5) {$\cdots$};
\draw (.3,.3) --node[right]{$\scriptstyle y_n$} (.3,.7);
\draw (.3,-.3) --node[right]{$\scriptstyle x_n$} (.3,-.7);
\roundNbox{}{(0,0)}{.3}{.3}{.3}{$f$}
}
\overset{\Gamma_\varphi}\longmapsto 
\prod_{j=1}^n \delta_{y_j=a_j} \left(\frac{d_{b_j}}{d_{a_j}}\right)^{1/4}
\tikzmath{
\draw (-.3,.3) --node[left]{$\scriptstyle a_1$} (-.3,.7);
\draw (-.3,-.3) --node[left]{$\scriptstyle x_1$} (-.3,-.7);
\node at (.05,.5) {$\cdots$};
\node at (.05,-.5) {$\cdots$};
\draw (.3,.3) --node[right]{$\scriptstyle a_n$} (.3,.7);
\draw (.3,-.3) --node[right]{$\scriptstyle x_n$} (.3,-.7);
\draw (-.3,1.3) --node[left]{$\scriptstyle b_1$} (-.3,1.7);
\node at (.05,1.5) {$\cdots$};
\draw (.3,1.3) --node[right]{$\scriptstyle b_n$} (.3,1.7);
\roundNbox{}{(0,0)}{.3}{.3}{.3}{$f$}
\roundNbox{}{(0,1)}{.3}{.3}{.3}{$\varphi$}
}
\]
The fourth roots above are needed to ensure that $\Gamma_\varphi^\dag = \Gamma_{\varphi^\dag}$ \cite[Eq.~(2)]{2305.14068}, i.e., $\Gamma: \fB_n\to B(\cH_n)$ is a unital $*$-algebra map.
We also have a precomposition gluing operator $\widetilde{\Gamma}_\varphi$ on $\cH_n$ given by
\[
\tikzmath{
\draw (-.3,.3) --node[left]{$\scriptstyle y_1$} (-.3,.7);
\draw (-.3,-.3) --node[left]{$\scriptstyle x_1$} (-.3,-.7);
\node at (.05,.5) {$\cdots$};
\node at (.05,-.5) {$\cdots$};
\draw (.3,.3) --node[right]{$\scriptstyle y_n$} (.3,.7);
\draw (.3,-.3) --node[right]{$\scriptstyle x_n$} (.3,-.7);
\roundNbox{}{(0,0)}{.3}{.3}{.3}{$f$}
}
\overset{\widetilde{\Gamma}_\varphi}\longmapsto 
\prod_{j=1}^n \delta_{x_j=b_j} \left(\frac{d_{a_j}}{d_{b_j}}\right)^{1/4}
\tikzmath{
\draw (-.3,.3) --node[left]{$\scriptstyle b_1$} (-.3,.7);
\draw (-.3,-.3) --node[left]{$\scriptstyle a_1$} (-.3,-.7);
\node at (.05,.5) {$\cdots$};
\node at (.05,-.5) {$\cdots$};
\draw (.3,.3) --node[right]{$\scriptstyle b_n$} (.3,.7);
\draw (.3,-.3) --node[right]{$\scriptstyle a_n$} (.3,-.7);
\draw (-.3,1.3) --node[left]{$\scriptstyle y_1$} (-.3,1.7);
\node at (.05,1.5) {$\cdots$};
\draw (.3,1.3) --node[right]{$\scriptstyle y_n$} (.3,1.7);
\roundNbox{}{(0,0)}{.3}{.3}{.3}{$\varphi$}
\roundNbox{}{(0,1)}{.3}{.3}{.3}{$f$}
}\,
\]
where again, $\widetilde{\Gamma}_\varphi^\dag = \widetilde{\Gamma}_{\varphi^\dag}$.
If we let $J_n\coloneqq  \dag: \cH_n\to \cH_n$, which is an involution (a period 2 conjugate-linear unitary isomorphism), we compute that
\[
(J_n \widetilde{\Gamma}_{\varphi^\dag} J_n )
f
=
\prod_{j=1}^n \delta_{y_j=a_j} \left(\frac{d_{b_j}}{d_{a_j}}\right)^{1/4}
J_n\left(
\tikzmath{
\draw (-.3,.3) --node[left]{$\scriptstyle a_1$} (-.3,.7);
\draw (-.3,-.3) --node[left]{$\scriptstyle b_1$} (-.3,-.7);
\node at (.05,.5) {$\cdots$};
\node at (.05,-.5) {$\cdots$};
\draw (.3,.3) --node[right]{$\scriptstyle a_n$} (.3,.7);
\draw (.3,-.3) --node[right]{$\scriptstyle b_n$} (.3,-.7);
\draw (-.3,1.3) --node[left]{$\scriptstyle x_1$} (-.3,1.7);
\node at (.05,1.5) {$\cdots$};
\draw (.3,1.3) --node[right]{$\scriptstyle x_n$} (.3,1.7);
\roundNbox{}{(0,0)}{.3}{.3}{.3}{$\varphi^\dag$}
\roundNbox{}{(0,1)}{.3}{.3}{.3}{$f^\dag$}
}
\right)
=
\prod_{j=1}^n \delta_{y_j=a_j} \left(\frac{d_{b_j}}{d_{a_j}}\right)^{1/4}
\tikzmath{
\draw (-.3,.3) --node[left]{$\scriptstyle a_1$} (-.3,.7);
\draw (-.3,-.3) --node[left]{$\scriptstyle x_1$} (-.3,-.7);
\node at (.05,.5) {$\cdots$};
\node at (.05,-.5) {$\cdots$};
\draw (.3,.3) --node[right]{$\scriptstyle a_n$} (.3,.7);
\draw (.3,-.3) --node[right]{$\scriptstyle x_n$} (.3,-.7);
\draw (-.3,1.3) --node[left]{$\scriptstyle b_1$} (-.3,1.7);
\node at (.05,1.5) {$\cdots$};
\draw (.3,1.3) --node[right]{$\scriptstyle b_n$} (.3,1.7);
\roundNbox{}{(0,0)}{.3}{.3}{.3}{$f$}
\roundNbox{}{(0,1)}{.3}{.3}{.3}{$\varphi$}
}
=
\Gamma_\varphi f
.
\]
We can summarize the above calculation as follows, where we write $\Omega_n\in \cH_n$ to denote the cyclic and separating vector $\id_n\in \fB_n$.

\begin{lem}
\label{lem:SkeinModuleJOps}
The skein module $\cH_n$ with the positive cone
$P_n=\operatorname{span}_{\bbR_{\geq 0}}\set{\Gamma_\varphi \widetilde{\Gamma}_{\varphi^\dag}\Omega_n}{\varphi\in \fB_n}$
is a standard form for $\fB_n$
satisfying
$\displaystyle J_n\Gamma_{\varphi}^\dag J_n
=
\widetilde{\Gamma}_\varphi$.
\end{lem}

Observe that this standard form (Hilbert algebra) has non-trivial modular automorphism group:
\begin{equation}
\label{eq:ModularGroupOnSkeinModule}
\Gamma_\varphi \Omega_n
=
\left(\frac{d_{b_1}\cdots d_{b_n}}{d_{a_1}\cdots d_{a_n}}\right)^{1/4}
\tikzmath{
\draw (-.3,.3) --node[left]{$\scriptstyle a_1$} (-.3,.7);
\node at (.05,.5) {$\cdots$};
\draw (.3,.3) --node[right]{$\scriptstyle a_n$} (.3,.7);
\draw (-.3,1.3) --node[left]{$\scriptstyle b_1$} (-.3,1.7);
\node at (.05,1.5) {$\cdots$};
\draw (.3,1.3) --node[right]{$\scriptstyle b_n$} (.3,1.7);
\roundNbox{}{(0,1)}{.3}{.3}{.3}{$\varphi$}
}
=
\left(\frac{d_{b_1}\cdots d_{b_n}}{d_{a_1}\cdots d_{a_n}}\right)^{1/2}
\widetilde{\Gamma}_\varphi \Omega_n.
\end{equation}
We define a faithful weight (non-normalized positive linear functional) given on
$\varphi\in \cX(a_1\otimes \cdots \otimes a_n\to b_1\otimes \cdots \otimes b_n)$
by
\begin{equation}
\label{eq:SkeinModuleWeight}
\omega(\varphi)
\coloneqq 
\langle \Omega_n | \Gamma_\varphi\Omega_n\rangle_{\cH_n}
=
\left(\frac{d_{b_1}\cdots d_{b_n}}{d_{a_1}\cdots d_{a_n}}\right)^{1/4}
\cdot
\langle \Omega_n | \varphi\rangle_{\cH_n}
=
\prod_{j=1}^n \delta_{a_j=b_j} \frac{1}{d_{a_j}}\tr_\cX(\varphi).
\end{equation}
Observe that $\omega(\id_n) = |\Irr(\cX)|^n$.

We now calculate the action of the modular automorphism group $\sigma^\omega$ on $\fB_n$.
The adjoint of the operator
$S_\omega \varphi\coloneqq \varphi^\dag$
is determined by the following calculation, where
$\phi\in \cX(b_1\otimes \cdots \otimes b_n\to a_1\otimes \cdots \otimes a_n)$:
\[
\begin{aligned}
\langle S_\omega \varphi |\phi\rangle
&=
\langle \varphi^\dag |\phi\rangle
\\&=
\left(\frac{1}{d_{b_1}\cdots d_{b_n}}\right)
\cdot
\tr_\cX(\varphi\circ \phi)
\\&=
\left(\frac{d_{a_1}\cdots d_{a_n}}{d_{b_1}\cdots d_{b_n}}\right)
\cdot
\left(\frac{1}{d_{a_1}\cdots d_{a_n}}\right)
\cdot
\tr_\cX(\phi\circ \varphi)
\\&=
\left(\frac{d_{a_1}\cdots d_{a_n}}{d_{b_1}\cdots d_{b_n}}\right)
\cdot
\langle \phi^\dag| \varphi\rangle.
\end{aligned}
\qquad\Longrightarrow\qquad
S_\omega^*(\phi)
=
\left(\frac{d_{a_1}\cdots d_{a_n}}{d_{b_1}\cdots d_{b_n}}\right)
\cdot
\phi^\dag.
\]
Setting $\Delta_\omega\coloneqq S_\omega^*S_\omega$, we have
\[
\Delta_\omega \varphi = S_\omega^*S_\omega \varphi = S_\omega^* \varphi^\dag 
=
\left(\frac{d_{a_1}\cdots d_{a_n}}{d_{b_1}\cdots d_{b_n}}\right)
\cdot
\varphi,
\]
which gives the following formula for the modular automorphism group:
\begin{equation}
\label{eq:SkeinModuleWEightModularAutomorphismGroup}
\sigma^\omega_t(\varphi)
=
\Delta_\omega^{it} \varphi \Delta_\omega^{-it}
=
\left(\frac{d_{a_1}\cdots d_{a_n}}{d_{b_1}\cdots d_{b_n}}\right)^{it}
\varphi.
\end{equation}
Now using $J_\omega = S_\omega\Delta_\omega^{-1/2}$,  we have 
\[
J_\omega \varphi 
= 
\left(\frac{d_{a_1}\cdots d_{a_n}}{d_{b_1}\cdots d_{b_n}}\right)^{-1/2} 
\cdot 
\varphi^\dag
=
\left(\frac{d_{b_1}\cdots d_{b_n}}{d_{a_1}\cdots d_{a_n}}\right)^{1/2}
\cdot 
\varphi^\dag.
\]
We thus see that \eqref{eq:ModularGroupOnSkeinModule} witnesses the usual formula for the right action \cite[(2.5.9)]{UQSLbook};
for $\zeta: b_1\otimes\cdots \otimes b_n \to c_1\otimes \cdots\otimes c_n$,
\[
\zeta\lhd \varphi
\coloneqq 
J_\omega\varphi^\dag J_\omega \zeta
=
\left(\frac{d_{c_1}\cdots d_{c_n}}{d_{b_1}\cdots d_{b_n}}\right)^{1/2}
\cdot
J_\omega\varphi^\dag\zeta^\dag 
=
\left(\frac{d_{a_1}\cdots d_{a_n}}{d_{b_1}\cdots d_{b_n}}\right)^{1/2}
\cdot 
\zeta \varphi
=
\zeta\sigma_{-i/2}^\omega(\varphi).
\]

\begin{lem}
The unitary $u:L^2(\fB_n, \omega)\to \cH_n$ given by 
$\varphi\mapsto \Gamma_\varphi\Omega_n$ intertwines the left and right $\fB_n$-actions.
\end{lem}
\begin{proof}
For $\zeta: c_1\otimes\cdots \otimes c_n\to a_1\otimes\cdots\otimes a_n$,
\[
\Gamma_\varphi
u \zeta
=
\Gamma_\varphi
\Gamma_\zeta
\Omega_n
=
\left(\frac{d_{b_1}\cdots d_{b_n}}{d_{c_1}\cdots d_{c_n}}\right)^{1/4}
\cdot\varphi \zeta
=
\Gamma_{\varphi \zeta}\Omega_n
=
u(\varphi \zeta).
\]
Similarly, if $\zeta: b_1\otimes\cdots\otimes b_n \to c_1\otimes\cdots\otimes c_n$,
\begin{align*}
\widetilde{\Gamma}_\varphi u\zeta
&=
\widetilde{\Gamma}_\varphi \Gamma_\zeta
\Omega_n
=
\left(\frac{d_{a_1}\cdots d_{a_n}}{d_{b_1}\cdots d_{b_n}}\right)^{1/4}
\left(\frac{d_{c_1}\cdots d_{c_n}}{d_{b_1}\cdots d_{b_n}}\right)^{1/4}
\zeta\varphi
\\&=
\left(\frac{d_{a_1}\cdots d_{a_n}}{d_{b_1}\cdots d_{b_n}}\right)^{1/2}
\Gamma_{\zeta\varphi}\Omega_n
=
\left(\frac{d_{a_1}\cdots d_{a_n}}{d_{b_1}\cdots d_{b_n}}\right)^{1/2}
u\zeta\varphi
=
u(\zeta\sigma^\omega_{-i/2}(\varphi)).
\qedhere
\end{align*}
\end{proof}

\begin{rem}
The canonical state $\psi$ on the boundary algebra $\fB_n$ \cite[Prop.~5.5]{MR4945955} is given on
$\varphi\in \cX(a_1\otimes \cdots \otimes a_n\to b_1\otimes \cdots \otimes b_n)$ by
\[
\psi(\varphi)
\coloneqq 
\frac{1}{D_\cX^n}
\prod_{j=1}^n \delta_{a_j=b_j} d_{a_j} \tr_\cX(\varphi).
\]
Similar to the calculation above, 
\[
\Delta_\psi \varphi = \left(\frac{d_{b_1}\cdots d_{b_n}}{d_{a_1}\cdots d_{a_n}}\right)
\cdot
\varphi,
\]
giving the fomula (c.f.~\cite[Eq.~(5.7)]{MR4945955}\footnote{For the Tomita-Takesaki modular theory, it is standard to use $\beta=-1$ KMS states, whereas in \cite[Eq.~(5.7)]{MR4945955}, we worked with a $\beta=1$ KMS state.})

\begin{equation}
\label{eq:SigmaPsi}
\sigma^\psi_t(\varphi)
=
\Delta_\psi^{it}\varphi \Delta_\psi^{-it}
=
\left(\frac{d_{b_1}\cdots d_{b_n}}{d_{a_1}\cdots d_{a_n}}\right)^{it}
\cdot \varphi.
\end{equation}
Thus $\sigma^\psi_t = \sigma^\omega_{-t}$, i.e., $\psi$ and $\omega$ are in some sense `opposite' weights.\footnote{If $M$ is a von Neumann algebra with weight $\phi$, we can define $\phi^{\op}$ on $M^{\op}$ by $\phi^{\op}(x^{\op})=\phi(x)$.
Then $\sigma^{\phi^{\op}}_t(x^{\op})=(\sigma_{-t}^\phi(x))^{\op}$. 
Or in other words, a $\beta$-KMS state for $t \mapsto \sigma_t$ is obviously a $(-\beta)$-KMS state for $t \mapsto \sigma_{-t}$.}
\end{rem}

\subsection{The Levin-Wen model}
\label{sec:LevinWen}

We now consider the unitary tensor category version of the Levin-Wen model from \cite[\S2]{2305.14068}.
We work on a $\bbZ^2$ lattice, where each vertex carries the 4-valent skein module $\cH_v\coloneqq \cS_\cX(\bbD,4)\cong \End_\cX(X^{\otimes 2})$.
We have edge operators $A_\ell$ for each edge $\ell$ of our lattice $\cL$ that glue the skein module along the edge $\ell$.
Note that the operators $A_\ell$ enforce that the labels along $\ell$ match, or we get zero.
For a region $R$, we write $p^A_R$ for the projection onto the image of $\prod_{\ell\in R} A_\ell$.

Given a plaquette $p\subset R$, we have a plaquette operator $B_p\coloneqq \frac{1}{D_\cX}\sum_{x\in\Irr(\cX)} d_x \cdot B_p^x$ defined on the image of $\prod_{\ell\in p} A_\ell$.
The operator $B_p$ is a weighted sum of the $B_p^x$ that glue in a loop labeled by the simple object $x$ into the punctured skein module corresponding to the image of $\prod_{\ell\in p} A_\ell$.
For a region $R$, we write $p_R^B=\prod_{p\in R} B_p$ and finally
$$
p_R\coloneqq p_R^Bp_R^A.
$$
Then $(p_R)$ is a net of projections of a quantum spin system that satisfies the LTO axioms \cite[Thm.~4.8]{MR4945955}.
Moreover, for a disk-like region $R$, $\im(p_R)$ is isomorphic to the skein-module $\cS_\cX(\bbD,|\partial R|)$, where $|\partial R|$ denotes the number of edges that intersect $\partial R$ \cite[Thm.~2.9]{2305.14068} (see also \cite{MR3204497}).
We refer the reader to \cite[\S4]{MR4945955} for further details.

We now equip our $\bbZ^2$ lattice $\cL$ with a $\bbZ/2$-symmetry $\theta$ that reflects about a 1D hyperplane $\cK$.

\begin{assume}
\label{assume:yAxisCut}
For simplicty, we will consider our $\bbZ^2$ lattice as a subset of $\bbR^2$, offset by $(1/2,0)$, and our hyperplane $\cK$ will be the $y$-axis in $\bbR^2$.    
\end{assume}

The conjugate linear unitary $\theta$ maps $\cH_{(x,y)}$ to $\cH_{(-x,y)}$ by
\[
\cH_{(x,y)}
\ni
\tikzmath{
\draw (-.7,0) --node[below]{$\scriptstyle a$} (-.3,0);
\draw (.7,0) --node[above]{$\scriptstyle d$} (.3,0);
\draw (0,-.7) --node[left]{$\scriptstyle b$} (0,-.3);
\draw (0,.7) --node[right]{$\scriptstyle c$} (0,.3);
\draw[thin, cyan] (-.7,.7) -- (.7,-.7);
\draw[->, cyan, thin] (-.6,.4) -- (-.4,.6);
\draw[<-, cyan, thin] (.6,-.4) -- (.4,-.6);
\roundNbox{fill=white}{(0,0)}{.3}{0}{0}{$f$}
}
\coloneqq 
\tikzmath{
\draw (-.2,.3) --node[left]{$\scriptstyle c$} (-.2,.7);
\draw (.2,.3) --node[right]{$\scriptstyle d$} (.2,.7);
\draw (-.2,-.3) --node[left]{$\scriptstyle a$} (-.2,-.7);
\draw (.2,-.3) --node[right]{$\scriptstyle b$} (.2,-.7);
\roundNbox{}{(0,0)}{.3}{.1}{.1}{$f$}
}
\qquad\longmapsto\qquad
\tikzmath{
\draw (-.7,0) --node[below]{$\scriptstyle d$} (-.3,0);
\draw (.7,0) --node[above]{$\scriptstyle a$} (.3,0);
\draw (0,-.7) --node[left]{$\scriptstyle \overline{b}$} (0,-.3);
\draw (0,.7) --node[right]{$\scriptstyle \overline{c}$} (0,.3);
\draw[thin, cyan] (-.7,.7) -- (.7,-.7);
\draw[->, cyan, thin] (-.6,.4) -- (-.4,.6);
\draw[<-, cyan, thin] (.6,-.4) -- (.4,-.6);
\roundNbox{fill=white}{(0,0)}{.3}{0}{0}{$\theta f$}
}
\coloneqq 
\tikzmath{
\draw (-.2,.3) --node[right]{$\scriptstyle a$} (-.2,.7);
\draw (.2,.3) arc(180:0:.2cm) --node[right]{$\scriptstyle \overline{b}$} (.6,-.7);
\draw (-.2,-.3) arc(0:-180:.2cm) --node[left]{$\scriptstyle \overline{c}$} (-.6,.7);
\draw (.2,-.3) --node[left]{$\scriptstyle d$} (.2,-.7);
\roundNbox{}{(0,0)}{.3}{.1}{.1}{$f^\dag$}
}
\in \cH_{(-x,y)}\,.
\]

Reflection positivity was first proved for the Levin-Wen model in \cite{MR4132814}.
The following lemma, whose proof is straightforward and left to the reader, shows that this version of the Levin-Wen model is reflection positive in the sense of \cite[Def.~4.1]{2510.20662}.

\begin{lem}
For every edge $\ell$ and plaquette $p$ in $\cL$,
$\theta A_{\ell} \theta = A_{\theta \ell}$ and $\theta B_{p} \theta = B_{\theta p}$.
Hence $p_{\theta R}=p_R$ for every disk-like region $R$.
\end{lem}

Due to an annoying clash of convention, we will write $\bbH_+$ for the half space with \emph{negative} $x$-values and $\bbH_-$ for the half-space with \emph{positive} $x$-values.
Given a $\theta$-symmetric disk-like region $R\subset \cL$, we write $R_\pm\coloneqq R\cap \bbH_\pm$, and we observe that $R_\pm$ is always disk-like.
The reason for swapping the typical $\pm$ convention here is that on the local ground state space $\im(p_R)$, which is isomorphic to the skein module $\cS_\cX(\bbD,|\partial R|)$,
the gluing operators $\Gamma_\varphi$ for $\varphi\in \cX(a_1\otimes\cdots \otimes a_n \to b_1\otimes\cdots\otimes b_n)$ along the hyperplane $\cK$ act on $R_+$ by post-composition, whereas on $R_-$ along $\cK$, the gluing operator $\widetilde{\Gamma}_\varphi$ acts by pre-composition.

Similarly to Lemma \ref{lem:SkeinModuleJOps}, when acting on sites along $\cK$, 
one can show that $\theta \Gamma_\varphi^\dag \theta = \widetilde{\Gamma}_\varphi$, where $\Gamma_\varphi$ acts on the right-hand side (by post-composition) of the half-plane $\bbH_+$ on the left of $\cK$, and  $\widetilde{\Gamma}_\varphi$ acts on the left-hand side (by pre-composition) on the half-plane $\bbH_-$ on the right of $\cK$. 
(The proof in the next subsection illustates these conventions.)
We omit the proof which is a straightforward exercise using the graphical calculus.

\subsection{The proof of \ref{LTO:RP} for Levin-Wen models}

We now show \ref{LTO:RP} for the Levin-Wen model.
We start with a ground state $|\Omega\rangle$, for which $|\Omega\rangle= \textcolor{red}{p_R}|\Omega\rangle$.
We consider the case where $R \ll_1 S$, and $I=\partial R \cap S$ is $n$ sites along a vertical slice in our Levin-Wen lattice.
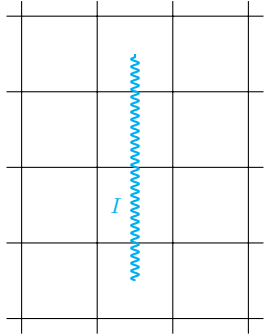
\[
\tikzmath{
\draw (-.2,-.2) grid (3.2,4.2);
\draw[thick, cyan, snake] (1.5,.5) -- (1.5,3.5);
\node[cyan, scale=.75] at (1.25,1.5) {$I$};
}
\]
In order to apply the gluing operator $\Gamma_\varphi$, we first use \eqref{eq:Fusion} to resolve all plaquette operators that act on the sites in $I$.
In the diagrams below, to ease the notation, we suppress all unnecessary sums over simples and scalars, only keeping track of labels, sums, and scalars for sites in $I$, which are marked in cyan in the diagram on the left.
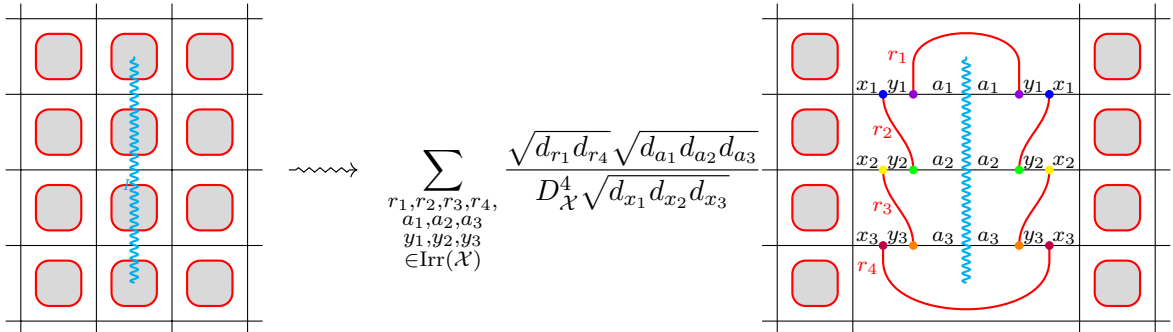
\[
\tikzmath{
\draw (-.2,-.2) grid (3.2,4.2);
\foreach \x in {0,1,2}{
\foreach \y in {0,1,2,3}{
\filldraw[red, thick, fill=gray!30, rounded corners=5pt] ($ (\x,\y) + (.2,.2) $) rectangle ($ (\x,\y) + (.8,.8) $);
}}
\draw[thick, cyan, snake] (1.5,.5) -- (1.5,3.5);
\node[cyan, scale=.75] at (1.4,1.8) {$\scriptstyle I$};
}
\begin{tikzcd}
\mbox{}
\arrow[r,squiggly]
&
\mbox{}
\end{tikzcd}
\sum_{\substack{
r_1,r_2,r_3,r_4,
\\
a_1,a_2,a_3
\\
y_1,y_2,y_3 
\\
\in \Irr(\cX)
}}
\frac{\sqrt{d_{r_1}d_{r_4}}\sqrt{d_{a_1}d_{a_2}d_{a_3}}}{D_\cX^{4}\sqrt{d_{x_1}d_{x_2}d_{x_3}}}
\tikzmath{
\draw (-.2,4) -- (5.2,4);
\foreach \x in {0,1,4,5}{
\draw (\x,-.2) -- (\x,4.2);
}
\foreach \y in {0,1,2,3}{
\draw (-.2,\y) -- (5.2,\y);
\foreach \x in {0,4}{
\filldraw[red, thick, fill=gray!30, rounded corners=5pt] ($ (\x,\y) + (.2,.2) $) rectangle ($ (\x,\y) + (.8,.8) $);
}
}
\foreach \y in {1,2,3}{
\node at ($ (1.2,4-\y) + (0,.1) $) {$\scriptstyle x_{\y}$};
\node at ($ (1.6,4-\y) + (0,.1) $) {$\scriptstyle y_{\y}$};
\node at ($ (2.2,4-\y) + (0,.1) $) {$\scriptstyle a_{\y}$};
\node at ($ (2.8,4-\y) + (0,.1) $) {$\scriptstyle a_{\y}$};
\node at ($ (3.4,4-\y) + (0,.1) $) {$\scriptstyle y_{\y}$};
\node at ($ (3.8,4-\y) + (0,.1) $) {$\scriptstyle x_{\y}$};
}
\draw[red, thick] (1.4,1)  -- (1.4,.8) to[out=-90,in=-90] (3.6,.8) -- (3.6,1);
\draw[red, thick] (1.8,3)  -- (1.8,3.4) to[out=90,in=90] (3.2,3.4) -- (3.2,3);
\draw[thick, red] (1.4,3) to[out=-90,in=90] (1.8,2);
\draw[thick, red] (1.4,2) to[out=-90,in=90] (1.8,1);
\draw[thick, red] (3.6,2) to[out=-90,in=90] (3.2,1);
\draw[thick, red] (3.6,3) to[out=-90,in=90] (3.2,2);
\node[red] at (1.6,3.5) {$\scriptstyle r_1$};
\node[red] at (1.4,2.5) {$\scriptstyle r_2$};
\node[red] at (1.4,1.5) {$\scriptstyle r_3$};
\node[red] at (1.2,.7) {$\scriptstyle r_4$};
\filldraw[purple] (1.4,1) circle (.05cm);
\filldraw[purple] (3.6,1) circle (.05cm);
\filldraw[orange] (1.8,1) circle (.05cm);
\filldraw[orange] (3.2,1) circle (.05cm);
\filldraw[yellow] (1.4,2) circle (.05cm);
\filldraw[yellow] (3.6,2) circle (.05cm);
\filldraw[green] (1.8,2) circle (.05cm);
\filldraw[green] (3.2,2) circle (.05cm);
\filldraw[blue] (1.4,3) circle (.05cm);
\filldraw[blue] (3.6,3) circle (.05cm);
\filldraw[violet] (1.8,3) circle (.05cm);
\filldraw[violet] (3.2,3) circle (.05cm);
\draw[thick, cyan, snake] (2.5,.5) -- (2.5,3.5);
}
\]
We now apply $\Gamma_\varphi$ at $I$; ignoring all regions outside the plaquettes containing $I$, we obtain
\[
\left(
\frac{d_{b_1}d_{b_2}d_{b_3}}{d_{a_1}d_{a_2}d_{a_3}}
\right)^{1/4}
\frac{\sqrt{d_{a_1}d_{a_2}d_{a_3}}}{D_\cX^{4}\sqrt{d_{x_1}d_{x_2}d_{x_3}}}
\sum_{\substack{
r_1,r_2,r_3,r_4,
\\
y_1,y_2,y_3 
\\
\in \Irr(\cX)
}}
\sqrt{d_{r_1}d_{r_4}}
\tikzmath{
\draw (.8,4) -- (5.2,4);
\foreach \x in {1,5}{
\draw (\x,-.2) -- (\x,4.2);
}
\foreach \y in {0,1,2,3}{
\draw (.8,\y) -- (5.2,\y);
}
\foreach \y in {1,2,3}{
\node at ($ (1.2,4-\y) + (0,.1) $) {$\scriptstyle x_{\y}$};
\node at ($ (1.6,4-\y) + (0,.1) $) {$\scriptstyle y_{\y}$};
\node at ($ (2,4-\y) + (0,.1) $) {$\scriptstyle a_{\y}$};
\node at ($ (3,4-\y) + (0,.15) $) {$\scriptstyle b_{\y}$};
\node at ($ (4,4-\y) + (0,.1) $) {$\scriptstyle a_{\y}$};
\node at ($ (4.4,4-\y) + (0,.1) $) {$\scriptstyle y_{\y}$};
\node at ($ (4.8,4-\y) + (0,.1) $) {$\scriptstyle x_{\y}$};
}
\draw[red, thick] (1.4,1) arc(-180:0:1.6 and .8) (4.6,1);
\draw[red, thick] (1.8,3) to[out=90,in=90] (4.2,3);
\draw[thick, red] (1.4,3) to[out=-90,in=90] (1.8,2);
\draw[thick, red] (1.4,2) to[out=-90,in=90] (1.8,1);
\draw[thick, red] (4.6,2) to[out=-90,in=90] (4.2,1);
\draw[thick, red] (4.6,3) to[out=-90,in=90] (4.2,2);
\filldraw[fill=white, rounded corners=5pt, very thick] (2.2,.7) rectangle (2.8,3.3);
\node at (2.5,2) {$\varphi$};
\node[red] at (1.8,3.5) {$\scriptstyle r_1$};
\node[red] at (1.4,2.5) {$\scriptstyle r_2$};
\node[red] at (1.4,1.5) {$\scriptstyle r_3$};
\node[red] at (1.3,.7) {$\scriptstyle r_4$};
\filldraw[purple] (1.4,1) circle (.05cm);
\filldraw[purple] (4.6,1) circle (.05cm);
\filldraw[orange] (1.8,1) circle (.05cm);
\filldraw[orange] (4.2,1) circle (.05cm);
\filldraw[yellow] (1.4,2) circle (.05cm);
\filldraw[yellow] (4.6,2) circle (.05cm);
\filldraw[green] (1.8,2) circle (.05cm);
\filldraw[green] (4.2,2) circle (.05cm);
\filldraw[blue] (1.4,3) circle (.05cm);
\filldraw[blue] (4.6,3) circle (.05cm);
\filldraw[violet] (1.8,3) circle (.05cm);
\filldraw[violet] (4.2,3) circle (.05cm);
\draw[thick, cyan, snake, knot] (3.5,.5) -- (3.5,3.5);
}
\]
It suffices to evaluate the above diagram in the annular skein module \cite[\S2.4]{2305.14068}, viewing the interval $I$ as a puncture.
The calculation is now similar to \cite[p27-9]{MR4945955}.
We contract along the $r_1$ strand using the fusion relation \eqref{eq:Fusion} to obtain
$$
\left(
\frac{d_{b_1}d_{b_2}d_{b_3}}{d_{a_1}d_{a_2}d_{a_3}}
\right)^{1/4}
d_{a_1}
\frac{\sqrt{d_{a_2}d_{a_3}}}{D_\cX^{4}\sqrt{d_{x_1}d_{x_2}d_{x_3}}}
\sum_{\substack{
r_2,r_3,r_4,
\\
y_1,y_2,y_3 
\\
\in \Irr(\cX)
}}
\sqrt{d_{y_1}d_{r_4}}
\tikzmath{
\draw (.8,3) --node[above,yshift=-.1cm]{$\scriptstyle x_1$} (1.4,3) to[out=0,in=180] (2.1,3.8) --node[above,yshift=-.1cm]{$\scriptstyle y_1$} (3.9,3.8) to[out=0,in=180] (4.6,3) --node[above,yshift=-.1cm]{$\scriptstyle x_1$} (5.2,3);
\draw (3.8,3) --node[above, yshift=-.1cm]{$\scriptstyle b_1$} (2.2,3) node[left,yshift=-.1cm]{$\scriptstyle a_1$} arc (270:90:.3cm) -- (3.8,3.6) arc (90:-90:.3cm) node[right,yshift=-.1cm]{$\scriptstyle a_1$};
\foreach \y in {1,2}{
\draw (.8,\y) -- (5.2,\y);
}
\foreach \y in {2,3}{
\node at ($ (1.2,4-\y) + (0,.1) $) {$\scriptstyle x_{\y}$};
\node at ($ (1.6,4-\y) + (0,.1) $) {$\scriptstyle y_{\y}$};
\node at ($ (2,4-\y) + (0,.1) $) {$\scriptstyle a_{\y}$};
\node at ($ (3,4-\y) + (0,.15) $) {$\scriptstyle b_{\y}$};
\node at ($ (4,4-\y) + (0,.1) $) {$\scriptstyle a_{\y}$};
\node at ($ (4.4,4-\y) + (0,.1) $) {$\scriptstyle y_{\y}$};
\node at ($ (4.8,4-\y) + (0,.1) $) {$\scriptstyle x_{\y}$};
}
\draw[red, thick] (1.4,1) arc(-180:0:1.6 and .8) (4.6,1);
\draw[thick, red] (1.4,3) to[out=-90,in=90] (1.8,2);
\draw[thick, red] (1.4,2) to[out=-90,in=90] (1.8,1);
\draw[thick, red] (4.6,2) to[out=-90,in=90] (4.2,1);
\draw[thick, red] (4.6,3) to[out=-90,in=90] (4.2,2);
\filldraw[fill=white, rounded corners=5pt, very thick] (2.2,.7) rectangle (2.8,3.3);
\node at (2.5,2) {$\varphi$};
\node[red] at (1.4,2.5) {$\scriptstyle r_2$};
\node[red] at (1.4,1.5) {$\scriptstyle r_3$};
\node[red] at (1.3,.7) {$\scriptstyle r_4$};
\filldraw[purple] (1.4,1) circle (.05cm);
\filldraw[purple] (4.6,1) circle (.05cm);
\filldraw[orange] (1.8,1) circle (.05cm);
\filldraw[orange] (4.2,1) circle (.05cm);
\filldraw[yellow] (1.4,2) circle (.05cm);
\filldraw[yellow] (4.6,2) circle (.05cm);
\filldraw[green] (1.8,2) circle (.05cm);
\filldraw[green] (4.2,2) circle (.05cm);
\filldraw[blue] (1.4,3) circle (.05cm);
\filldraw[blue] (4.6,3) circle (.05cm);
\draw[thick, cyan, snake, knot] (3.5,.5) -- (3.5,3.5);
}\,.
$$
Next, we use the fusion relation \eqref{eq:Fusion} to contract along the $y_1$ string to obtain
$$
\left(
\frac{d_{b_1}d_{b_2}d_{b_3}}{d_{a_1}d_{a_2}d_{a_3}}
\right)^{1/4}
d_{a_1}
\frac{\sqrt{d_{a_2}d_{a_3}}}{D_\cX^{4}\sqrt{d_{x_2}d_{x_3}}}
\sum_{\substack{
r_2,r_3,r_4,
\\
y_2,y_3 
\\
\in \Irr(\cX)
}}
\sqrt{d_{r_2}d_{r_4}}
\tikzmath{
\draw (.8,3) -- (1,3) to[out=0,in=180] (1.7,4) --node[above,yshift=-.1cm]{$\scriptstyle x_1$} (4.3,4) to[out=0,in=180] (5,3) -- (5.2,3);
\draw[thick, red] (1.8,2) to[out=90,in=-90] (1.6,3.3) arc (180:90:.5cm) -- (3.9,3.8) arc (90:0:.5cm) to[out=-90,in=90] (4.2,2);
\node[red] at (1.5,2.5) {$\scriptstyle r_2$};
\draw (3.8,3) --node[above, yshift=-.1cm]{$\scriptstyle b_1$} (2.2,3) node[left,yshift=-.1cm]{$\scriptstyle a_1$} arc (270:90:.3cm) -- (3.8,3.6) arc (90:-90:.3cm) node[right,yshift=-.1cm]{$\scriptstyle a_1$};
\foreach \y in {1,2}{
\draw (.8,\y) -- (5.2,\y);
}
\foreach \y in {2,3}{
\node at ($ (1.2,4-\y) + (0,.1) $) {$\scriptstyle x_{\y}$};
\node at ($ (1.6,4-\y) + (0,.1) $) {$\scriptstyle y_{\y}$};
\node at ($ (2,4-\y) + (0,.1) $) {$\scriptstyle a_{\y}$};
\node at ($ (3,4-\y) + (0,.15) $) {$\scriptstyle b_{\y}$};
\node at ($ (4,4-\y) + (0,.1) $) {$\scriptstyle a_{\y}$};
\node at ($ (4.4,4-\y) + (0,.1) $) {$\scriptstyle y_{\y}$};
\node at ($ (4.8,4-\y) + (0,.1) $) {$\scriptstyle x_{\y}$};
}
\draw[red, thick] (1.4,1) arc(-180:0:1.6 and .8) (4.6,1);
\draw[thick, red] (1.4,2) to[out=-90,in=90] (1.8,1);
\draw[thick, red] (4.6,2) to[out=-90,in=90] (4.2,1);
\filldraw[fill=white, rounded corners=5pt, very thick] (2.2,.7) rectangle (2.8,3.3);
\node at (2.5,2) {$\varphi$};
\node[red] at (1.4,1.5) {$\scriptstyle r_3$};
\node[red] at (1.3,.7) {$\scriptstyle r_4$};
\filldraw[purple] (1.4,1) circle (.05cm);
\filldraw[purple] (4.6,1) circle (.05cm);
\filldraw[orange] (1.8,1) circle (.05cm);
\filldraw[orange] (4.2,1) circle (.05cm);
\filldraw[yellow] (1.4,2) circle (.05cm);
\filldraw[yellow] (4.6,2) circle (.05cm);
\filldraw[green] (1.8,2) circle (.05cm);
\filldraw[green] (4.2,2) circle (.05cm);
\draw[thick, cyan, snake, knot] (3.5,.5) -- (3.5,3.5);
}\,.
$$
We then use the fusion relation \eqref{eq:Fusion} to contract along the $r_2$ string to obtain
$$
\left(
\frac{d_{b_1}d_{b_2}d_{b_3}}{d_{a_1}d_{a_2}d_{a_3}}
\right)^{1/4}
d_{a_1}d_{a_2}
\frac{\sqrt{d_{a_3}}}{D_\cX^{4}\sqrt{d_{x_2}d_{x_3}}}
\sum_{\substack{
r_3,r_4,
\\
y_2,y_3 
\\
\in \Irr(\cX)
}}
\sqrt{d_{y_2}d_{r_4}}
\tikzmath{
\draw (.8,3) to[out=0,in=180] (1.7,4.2) --node[above,yshift=-.1cm]{$\scriptstyle x_1$} (4.3,4.2) to[out=0,in=180] (5.2,3);
\draw (3.8,3) --node[above, yshift=-.1cm]{$\scriptstyle b_1$} (2.2,3) node[left,yshift=-.1cm]{$\scriptstyle a_1$} arc (270:90:.3cm) -- (3.8,3.6) arc (90:-90:.3cm) node[right,yshift=-.1cm]{$\scriptstyle a_1$};
\draw (.8,2) --node[above,yshift=-.1cm]{$\scriptstyle x_2$} (1.2,2) to[out=0,in=-90] node[left]{$\scriptstyle y_2$} (1.3,3) to[out=90,in=180] (1.9,4) -- (4.1,4) to[out=0,in=90] (4.7,3) to[out=-90,in=180] (4.8,2) --node[above,yshift=-.1cm]{$\scriptstyle x_2$} (5.2,2);
\draw (3.8,2) --node[above, yshift=-.1cm]{$\scriptstyle b_2$} (2.2,2) node[left,yshift=-.1cm]{$\scriptstyle a_2$} to[out=180,in=-90] (1.6,3.3) arc (180:90:.5cm) -- (3.9,3.8) arc (90:0:.5cm) to[out=-90,in=0] (3.8,2) node[right,yshift=-.1cm]{$\scriptstyle a_2$};
\draw (.8,1) -- (5.2,1);
\node at (1.2,1.1) {$\scriptstyle x_3$};
\node at (1.6,1.1) {$\scriptstyle y_3$};
\node at (2,1.1) {$\scriptstyle a_3$};
\node at (3,1.15) {$\scriptstyle b_3$};
\node at (4,1.1) {$\scriptstyle a_3$};
\node at (4.4,1.1) {$\scriptstyle y_3$};
\node at (4.8,1.1) {$\scriptstyle x_3$};
\draw[red, thick] (1.4,1) arc(-180:0:1.6 and .8) (4.6,1);
\draw[thick, red] (1.2,2) to[out=-90,in=90] (1.8,1);
\draw[thick, red] (4.8,2) to[out=-90,in=90] (4.2,1);
\filldraw[fill=white, rounded corners=5pt, very thick] (2.2,.7) rectangle (2.8,3.3);
\node at (2.5,2) {$\varphi$};
\node[red] at (1.3,1.5) {$\scriptstyle r_3$};
\node[red] at (1.2,.7) {$\scriptstyle r_4$};
\filldraw[purple] (1.4,1) circle (.05cm);
\filldraw[purple] (4.6,1) circle (.05cm);
\filldraw[orange] (1.8,1) circle (.05cm);
\filldraw[orange] (4.2,1) circle (.05cm);
\filldraw[yellow] (1.2,2) circle (.05cm);
\filldraw[yellow] (4.8,2) circle (.05cm);
\draw[thick, cyan, snake, knot] (3.5,.5) -- (3.5,3.5);
}\,.
$$
At this point, it is clear that we can then contract along the $y_2$ string using \eqref{eq:Fusion}, followed by the $r_3$ string afterward to obtain 
\[
\left(
\frac{d_{b_1}d_{b_2}d_{b_3}}{d_{a_1}d_{a_2}d_{a_3}}
\right)^{1/4}
\frac{d_{a_1}d_{a_2}d_{a_3}}{D_\cX^{4}\sqrt{d_{x_3}}}
\sum_{\substack{r_4,y_3\\\in \Irr(\cX)}}
\sqrt{d_{y_3}d_{r_4}}
\tikzmath{
\draw (1.3,.3) --node[above, yshift=-.1cm]{$\scriptstyle b_1$} (-.3,.3) arc(270:90:.3cm) -- (1.3,.9) arc(90:-90:.3cm) node[above, yshift=-.1cm]{$\scriptstyle a_1$};
\draw (1.3,0) --node[above, yshift=-.1cm]{$\scriptstyle b_2$} (-.3,0) arc(270:90:.6cm) -- (1.3,1.2) arc(90:-90:.6cm) node[above, yshift=-.1cm]{$\scriptstyle a_2$};
\draw (1.3,-.3) --node[above, yshift=-.1cm]{$\scriptstyle b_3$} (-.3,-.3) arc(270:90:.9cm) -- (1.3,1.5) arc(90:-90:.9cm) node[above, yshift=-.1cm]{$\scriptstyle a_3$};
\draw (-3,-.3) --node[above, yshift=-.1cm]{$\scriptstyle x_3$} (-2.6,-.3) node[above, xshift=.2cm, yshift=-.1cm]{$\scriptstyle y_3$} to[out=0,in=180] (-.3,1.8) -- (1.3,1.8) to[out=0,in=180] (3.6,-.3) node[above, xshift=-.2cm, yshift=-.1cm]{$\scriptstyle y_3$} --node[above, yshift=-.1cm]{$\scriptstyle x_3$} (4,-.3);
\draw (-3,0) --node[above, yshift=-.1cm]{$\scriptstyle x_2$} (-2.8,0) to[out=0,in=180] (-.3,2.1) -- (1.3,2.1) to[out=0,in=180] (3.8,0) --node[above, yshift=-.1cm]{$\scriptstyle x_2$} (4,0);
\draw (-3,.3) node[above, xshift=.1cm, yshift=-.1cm]{$\scriptstyle x_1$} to[out=0,in=180] (-.3,2.4) -- (1.3,2.4) to[out=0,in=180] (4,.3) node[above, xshift=-.1cm, yshift=-.1cm]{$\scriptstyle x_1$};
\draw[thick, red] (-2.6,-.3) arc(-180:0:3.1cm and .8cm);
\node[red] at (-2.2,-.5) {$\scriptstyle r_4$};
\filldraw[fill=white, rounded corners=5pt, very thick] (-.3,-.6) rectangle (.3,.6);
\node at (0,0) {$\varphi$};
\filldraw[purple] (-2.6,-.3) circle (.05cm);
\filldraw[purple] (3.6,-.3) circle (.05cm);
\draw[thick, cyan, snake, knot] (1,-.6) -- (1,.6);
}
\]
and finally, contracting along the $y_3$ strand yields\footnote{Note that the excitation in \eqref{eq:GammaPhiAction} is manifestly topological;
it moves at no energy cost to other locations, as each localization is individually isomorphic to the skein module of the annulus.
Of course, this excitation is really a copy of the trivial tube algebra representation, as it was created locally by a boundary algebra operator.
Indeed, the interior of the $r_4$ strand that contains $\varphi$ should be viewed as a sphere, as the screening operator works from both sides. 
Thus the tube algebra representation is the cyclic module generated by $\varphi$, which is contained in the trivial representation.
}
\begin{equation}
\label{eq:GammaPhiAction}
\frac{(d_{a_1}d_{a_2}d_{a_3})^{3/4}(d_{b_1}d_{b_2}d_{b_3})^{1/4}}{D_\cX^{4}}
\sum_{r_4\in \Irr(\cX)}
d_{r_4}
\tikzmath{
\draw (1.3,.3) --node[above, yshift=-.1cm]{$\scriptstyle b_1$} (-.3,.3) arc(270:90:.3cm) -- (1.3,.9) arc(90:-90:.3cm) node[above, yshift=-.1cm]{$\scriptstyle a_1$};
\draw (1.3,0) --node[above, yshift=-.1cm]{$\scriptstyle b_2$} (-.3,0) arc(270:90:.6cm) -- (1.3,1.2) arc(90:-90:.6cm) node[above, yshift=-.1cm]{$\scriptstyle a_2$};
\draw (1.3,-.3) --node[above, yshift=-.1cm]{$\scriptstyle b_3$} (-.3,-.3) arc(270:90:.9cm) -- (1.3,1.5) arc(90:-90:.9cm) node[above, yshift=-.1cm]{$\scriptstyle a_3$};
\foreach \y in{1,2,3}{
\draw ($ (-2,3) - .3*(0,\y) $) --node[above, yshift=-.1cm]{$\scriptstyle x_\y$} ($ (3,3) - .3*(0,\y) $);
}
\node[red] at (-.6,-.6) {$\scriptstyle r_4$};
\draw[thick, red] (-.3,-.8) arc(270:90:1.3cm)  -- (1.3,1.8)  arc(90:-90:1.3cm) -- (-.3,-.8);
\filldraw[fill=white, rounded corners=5pt, very thick] (-.3,-.6) rectangle (.3,.6);
\node at (0,0) {$\varphi$};
\draw[thick, cyan, snake, knot] (1,-.6) -- (1,.6);
}
\end{equation}

Now we can also apply $\widetilde{\Gamma}_\varphi$ at $I$; ignoring regions outside the plaquettes containing $I$, we obtain
\[
\left(
\frac{d_{a_1}d_{a_2}d_{a_3}}{d_{b_1}d_{b_2}d_{b_3}}
\right)^{1/4}
\frac{\sqrt{d_{b_1}d_{b_2}d_{b_3}}}{D_\cX^{4}\sqrt{d_{x_1}d_{x_2}d_{x_3}}}
\sum_{\substack{
r_1,r_2,r_3,r_4,
\\
y_1,y_2,y_3 
\\
\in \Irr(\cX)
}}
\sqrt{d_{r_1}d_{r_4}}
\tikzmath{
\draw (.8,4) -- (5.2,4);
\foreach \x in {1,5}{
\draw (\x,-.2) -- (\x,4.2);
}
\foreach \y in {0,1,2,3}{
\draw (.8,\y) -- (5.2,\y);
}
\foreach \y in {1,2,3}{
\node at ($ (1.2,4-\y) + (0,.1) $) {$\scriptstyle x_{\y}$};
\node at ($ (1.6,4-\y) + (0,.1) $) {$\scriptstyle y_{\y}$};
\node at ($ (2,4-\y) + (0,.1) $) {$\scriptstyle b_{\y}$};
\node at ($ (3,4-\y) + (0,.15) $) {$\scriptstyle a_{\y}$};
\node at ($ (4,4-\y) + (0,.1) $) {$\scriptstyle b_{\y}$};
\node at ($ (4.4,4-\y) + (0,.1) $) {$\scriptstyle y_{\y}$};
\node at ($ (4.8,4-\y) + (0,.1) $) {$\scriptstyle x_{\y}$};
}
\draw[red, thick] (1.4,1) arc(-180:0:1.6 and .8) (4.6,1);
\draw[red, thick] (1.8,3) to[out=90,in=90] (4.2,3);
\draw[thick, red] (1.4,3) to[out=-90,in=90] (1.8,2);
\draw[thick, red] (1.4,2) to[out=-90,in=90] (1.8,1);
\draw[thick, red] (4.6,2) to[out=-90,in=90] (4.2,1);
\draw[thick, red] (4.6,3) to[out=-90,in=90] (4.2,2);
\filldraw[fill=white, rounded corners=5pt, very thick] (3.2,.7) rectangle (3.8,3.3);
\node at (3.5,2) {$\varphi$};
\node[red] at (1.8,3.5) {$\scriptstyle r_1$};
\node[red] at (1.4,2.5) {$\scriptstyle r_2$};
\node[red] at (1.4,1.5) {$\scriptstyle r_3$};
\node[red] at (1.3,.7) {$\scriptstyle r_4$};
\filldraw[purple] (1.4,1) circle (.05cm);
\filldraw[purple] (4.6,1) circle (.05cm);
\filldraw[orange] (1.8,1) circle (.05cm);
\filldraw[orange] (4.2,1) circle (.05cm);
\filldraw[yellow] (1.4,2) circle (.05cm);
\filldraw[yellow] (4.6,2) circle (.05cm);
\filldraw[green] (1.8,2) circle (.05cm);
\filldraw[green] (4.2,2) circle (.05cm);
\filldraw[blue] (1.4,3) circle (.05cm);
\filldraw[blue] (4.6,3) circle (.05cm);
\filldraw[violet] (1.8,3) circle (.05cm);
\filldraw[violet] (4.2,3) circle (.05cm);
\draw[thick, cyan, snake, knot] (2.5,.5) -- (2.5,3.5);
}
\]
By mimicking the above argument for $\Gamma_\varphi$, we see that this diagram simplifies to
\begin{equation}
\label{eq:GammaTildePhiAction}
\frac{(d_{a_1}d_{a_2}d_{a_3})^{1/4}(d_{b_1}d_{b_2}d_{b_3})^{3/4}}{D_\cX^{4}}
\sum_{\substack{r_4\\\in \Irr(\cX)}}
d_{r_4}
\tikzmath{
\draw (1.3,.3) --node[above, yshift=-.1cm]{$\scriptstyle a_1$}  (-.3,.3) node[above, yshift=-.1cm]{$\scriptstyle b_1$} arc(270:90:.3cm) -- (1.3,.9) arc(90:-90:.3cm);
\draw (1.3,0) --node[above, yshift=-.1cm]{$\scriptstyle a_2$}  (-.3,0) node[above, yshift=-.1cm]{$\scriptstyle b_2$} arc(270:90:.6cm) -- (1.3,1.2) arc(90:-90:.6cm);
\draw (1.3,-.3) --node[above, yshift=-.1cm]{$\scriptstyle a_3$}  (-.3,-.3) node[above, yshift=-.1cm]{$\scriptstyle b_3$} arc(270:90:.9cm) -- (1.3,1.5) arc(90:-90:.9cm);
\foreach \y in{1,2,3}{
\draw ($ (-2,3) - .3*(0,\y) $) --node[above, yshift=-.1cm]{$\scriptstyle x_\y$} ($ (3,3) - .3*(0,\y) $);
}
\node[red] at (-.6,-.6) {$\scriptstyle r_4$};
\draw[thick, red] (-.3,-.8) arc(270:90:1.3cm)  -- (1.3,1.8)  arc(90:-90:1.3cm) -- (-.3,-.8);
\filldraw[fill=white, rounded corners=5pt, very thick] (.7,-.6) rectangle (1.3,.6);
\node at (1,0) {$\varphi$};
\draw[thick, cyan, snake, knot] (0,-.6) -- (0,.6);
}
\end{equation}
Since
\[
\eqref{eq:GammaPhiAction}
=
\left(
\frac{d_{a_1}d_{a_2}d_{a_3}}{d_{b_1}d_{b_2}d_{b_3}}
\right)^{1/2}
\cdot
\eqref{eq:GammaTildePhiAction},
\]
we see that
\begin{align*}
(\Gamma_\varphi \otimes 1)p_\Delta
&=
\left(
\frac{d_{a_1}d_{a_2}d_{a_3}}{d_{b_1}d_{b_2}d_{b_3}}
\right)^{1/2}
(1\otimes \widetilde{\Gamma}_\varphi )p_\Delta
\\&
\underset{\eqref{eq:SigmaPsi}}{=}
(1\otimes\widetilde{\Gamma}_{\sigma^\psi_{i/2}(\varphi)})p_\Delta
=
(1\otimes\Theta(\Gamma_{\sigma^\psi_{i/2}(\varphi)}^\dag))p_\Delta
=
(1\otimes\Theta(\Gamma_{\sigma^\psi_{-i/2}(\varphi^\dag)}))p_\Delta.
\end{align*}
This concludes the proof of \ref{LTO:RP} for the Levin-Wen model.
\hfill\qed

\begin{rem}
\label{rem:thincone}
A similar proof works for arbitrary sufficiently large cone-like regions as in Construction \ref{construction:ConeLikeRegions} such that $\partial[\Lambda]$ is homeomorphic to $\bbR$.
\end{rem}

\begin{rem}
A similar proof as above shows that \ref{LTO:RP} also holds for the Walker-Wang model.
Indeed, one uses the description of the boundary algebra as the braided categorical net associated to a unitary braided fusion category from \cite[Thm.~5.6]{2506.19969}. 
\end{rem}

\subsection{Relating \ref{LTO:RP} to \texorpdfstring{\cite{2510.20662}}{[LZ25]}}
\label{sec:LTOvsRP}

In this section, we discuss how \ref{LTO:RP} compares with the reflection positivity setup of \cite{2510.20662}. 
We assume we are working with a quantum spin system satisfying the following assumptions, which are mostly identical to those used in \cite{2510.20662}.
First, we consider a family of interactions $\Phi(X) \in \fA(X)_+$ for each finite subset $X$.
For a finite subset $R$, we define the local Hamiltonian $H_R \coloneqq \sum_{X \subseteq R} \Phi(X)$. 
Following \cite{2510.20662}, we assume that the interactions satisfy the following conditions, which ensures that the local Hamiltonians are reflection positive:
\begin{itemize}
\item 
for every finite $X \subseteq \bbH_+$, $\Theta(\Phi(X)) = \Phi(\theta(X))$, and
\item 
for every finite $X$ such that $X_+ \coloneqq X \cap \bbH_+ \neq \emptyset$ and $X_- \coloneqq X \cap \bbH_- \neq \emptyset$, there exist linearly independent $x_j \in \fA(X_+)$ such that $\Phi(X) = \sum \Theta(x_j) \otimes x_j$.
\end{itemize}
We assume that the Hamiltonian is \emph{frustration free}, meaning that for every finite region $R$, $\ker H_R = \bigcap_{X \subseteq R} \ker \Phi(X)$ \cite{MR1158756, MR2742836, MR4700364}.
Note that in this case, the projections $p_R$ onto $\ker H_R$ form a net of projections.
Additionally, we assume that the interactions have \emph{finite range}, meaning that there exists $r > 0$ such that for every $X$ with radius at least $r$, $\Phi(X) = 0$. 
Finally, we assume that the net of projections $(p_\Lambda)$ is \emph{extendable}, meaning that for every $R \subset S$, $p_R$ is the range projection of $\Tr_{\cH_{S \setminus R}}(p_S)$. 
Note that this definition of extendability is slightly stronger than the one in \cite[Def.~5.18]{2510.20662}, but it is helpful for relating the approach in \cite{2510.20662} to our LTO approach.\footnote{Specifically, we are using the stronger condition to ensure that for a $\theta$-symmetric region $R$, the range projection of $\Tr_{\cH_{R_-}}(p_R)$ is $p_{R_+}$. 
In \cite{2510.20662}, the range projection of $\Tr_{\cH_{R_-}}(p_R)$ plays the role that $p_{R_+}$ does in the LTO approach.}

Let $R$ be a $\theta$-symmetric region and $R_+ \coloneqq R \cup \bbH_+$. 
Following \cite{MR1816609, 2510.20662}, we define the \emph{interaction algebra} $\cI(R) \subseteq \fA(R_+)$ as follows. 
We write $p_R = \sum x_j^- \otimes x_j^+$, where $x_j^{\pm} \in \fA(\Lambda_{\pm})$ and the $x_j^-$ are all linearly independent. 
Then $\cI(R) \coloneqq \rmC^*\{x_j^+\}$, which is independent of the choice of $x_j^{\pm}$ by \cite[Lem.~7]{MR2890305}. 
The net of boundary algebras computed in \cite{2510.20662} is isomorphic to $\cI(R) p_{R_+}$. 
In particular, if $R \subset S$, then $x \mapsto xp_{S_+}$ defines an injection $\cI(R) p_{R_+} \to \cI(S) p_{S_+}$ by \cite[Prop.~5.15]{2510.20662}, and the operators in this net are localized near the boundary since we are assuming the Hamiltonian has finite-range interactions \cite[Thm.~5.25]{2510.20662}. 
Furthermore, \cite{2510.20662} define a completely positive map $\scrE_R \colon \fA(R_+) \to \cI(R) p_{R_+}$ by 
\[
\scrE_R(x)
\coloneqq
\frac{1}{\Tr(p_R)}
\Tr_{\cH_-} (p_R(\Theta(x^\dag) \otimes 1)).
\]
Now, in the case that $R_+ \Subset S_+$, it is always true that $\cI(R) p_{S_+} \subseteq \fB(R_+ \Subset S_+)$.
However, for the LTO examples computed in \cite{MR4945955}, this inclusion is in fact an equality. 

So let us now assume that our model is also an LTO such that $\cI(R) p_{S_+} = \fB(R_+ \Subset S_+)$.
Recall that for $R_+ \Subset S_+$, we have a completely positive map $\bbE_{S_+} \colon \fA(R_+) \to \fB(R_+ \Subset S_+)$ given by $x \mapsto p_{S_+} x p_{S_+}$. 
However, $\bbE_{S_+}$ differs from $\scrE_S|_{\fA(R_+)}$ even when $\cI(R) p_{S_+} = \fB(R_+ \Subset S_+)$. 
Indeed, $\bbE_{S_+}(1) = p_{S_+}$, which is the unit of $\fB(R_+ \Subset S_+)$. 
On the other hand, \cite{2510.20662} showed that $\scrE_S(1) = \frac{1}{\Tr(p_S)} \Tr_{\cH_{S_-}}(p_S) = \Delta_\psi F$, where $F \in \cZ(\cI(S) p_{S_+})_+$.\footnote{Here, as we indentify $\Delta_\psi$ with the modular operator for the restriction of the state to finite regions, we are further assuming \ref{LTO:FaithfulOnBoundary} as in Corollary \ref{cor:ModularGroupPreservesSubalgebras}.} 
Assuming \ref{LTO:RP}, then we can relate $\scrE_S$ and $\bbE_{S_+}$ explicitly. 
Indeed, we have that 
\begin{align*}
\scrE_\Delta(x)
&=
\frac{1}{\Tr(p_S)} \Tr_{\cH_{S_-}} (p_S (\Theta(x^\dag) \otimes 1))
=
\frac{1}{\Tr(p_S)} \Tr_{\cH_{S_-}} (p_S (\Theta(p_{S_+}x^\dag p_{S_+}) \otimes 1))
\\&\underset{\text{\ref{LTO:RP}}}{=}
\frac{1}{\Tr(p_S)} \Tr_{\cH_{S_-}} (p_S (1 \otimes \sigma^\psi_{i/2}(p_{S_+}x p_{S_+})))
=
\frac{1}{\Tr(p_S)} \Tr_{\cH_{S_-}} (p_S)
\cdot
\sigma^\psi_{i/2}(p_{S_+}x p_{S_+})
\\&=
\Delta_\psi F
\cdot
\sigma^\psi_{i/2}(\bbE_{S_+}(x))
=
\sigma^\psi_{-i/2}(\bbE_{S_+}(x))
\cdot
\Delta_\psi F,
\end{align*}
where the last equality follows because $\sigma_t^\psi(x) = \Delta_\psi^{it} x \Delta_\psi^{-it}$ for $x \in \fB(R_+ \Subset S_+)$. 

Furthermore, \cite{2510.20662} compute an algebra $\cM_\Lambda$ using Osterwalder--Schrader (OS) reconstruction and show that $\cM_R \cong \cI(R) p_{R_+}$, although this isomorphism requires twisting by the modular operator. 
To compute $\cM_R$, \cite{2510.20662} first defines the sesquilinear form
\[
\langle x, y \rangle_R
\coloneqq
\frac{1}{\Tr(p_R)} \Tr(p_R (\Theta(y) \otimes x))
=
\Tr(\scrE_R(y^\dag)x).
\]
on $\fA(R_+)$, and $\scrH_R \coloneqq \fA(R_+)/\ker \langle - , - \rangle_R$.
For $x\in \fA(R_+)$, write $\hat x \in \scrH_R$ for its image under the quotient map.
When $x\in \fA(R_+)$ satisfies $x \ker \langle - , - \rangle_R \subseteq \ker \langle - , - \rangle_R$,
we get an operator $\Phi_R(x) \in B(\scrH_{R})$ given by $\Phi_R(x)\hat y = \widehat{xy}$.
The algebra given by OS reconstruction \cite[Def.~5.7]{2510.20662} is $\cM_R \coloneqq \im \Phi_R$.
For every $x \in \cM_R$, there is a unique $y \in \cI(R)p_{R_+}$ such that $x = \Phi_R(y)$ \cite[Prop.~5.10]{2510.20662}. 
Specifically, for $x = \Phi_R(w)$, we have that $y = p_{R_+} w p_{R_+}$.
However, $x^\dag = \Phi_R(\sigma^\psi_{i/2}(y^\dag))$.
Therefore, there is a $*$-isomorphism $\rho_R \colon \cI(R)p_{R_+} \to \cM_R$ given by $\rho(x) \coloneqq \Phi( \sigma_\psi^{i/4}(x))$ \cite[Thm.~5.11]{2510.20662}. 
In the case of the Levin--Wen model, $\cM_R = \End_\cC(X^n)$, where $n$ is the number of sites in $\partial R_+ \cap \partial \bbH_+$, and for $\varphi \in \End_\cC(X^n)$, $\Gamma_\varphi = \rho_R^{-1}(\varphi)$.

\bibliographystyle{alpha}
\bibliography{bibliography}

\end{document}